\documentclass{ijuc}
\usepackage[cp1251]{inputenc}
\usepackage{graphicx}
\usepackage{amsmath,systeme}
\usepackage{amssymb}
\usepackage{amsthm}

\usepackage{algorithm}
\usepackage[noend]{algpseudocode}

\usepackage{tikz}
\usetikzlibrary{arrows}
\tikzstyle{block}=[draw opacity=0.7,line width=1.4cm]

    \usepackage{graphicx}
    \usepackage{wrapfig}
     \usepackage[all]{xy}

    \usepackage{wrapfig}
     \usepackage[all]{xy}

\theoremstyle{plain}
\newtheorem{thm}{Theorem}[section]
\newtheorem{conj}{Conjecture}

\newtheorem{cons}{Corollary}[thm]
\newtheorem{lem}[thm]{Lemma}

\newtheorem{conslem}{Corollary}[thm]

\newtheorem*{THMExistsStrongSubalgebraTHM}{Theorem~\ref{ExistsStrongSubalgebraTHM}}
\newtheorem*{LEMCBTNonAbsorbing}{Lemma~\ref{CBTNonAbsorbing}}

\newtheorem*{LEMPCBsub}{Lemma~\ref{PCBsub}}

\newtheorem*{THMPCBint}{Theorem~\ref{PCBint}}

\newtheorem*{LEMNoEssential}{Lemma~\ref{NoEssential}}

\newtheorem*{THMCommonPropertiesThm}{Theorem~\ref{CommonPropertiesThm}}
\newtheorem*{LEMProjectiveSubCharacterization}{Lemma~\ref{ProjectiveSubCharacterization}}

\theoremstyle{remark}

\numberwithin{equation}{section}

\newcommand{\algA}{\mathbf A}
\newcommand{\algB}{\mathbf B}
\newcommand{\alg}[1]{\mathbf #1}

\DeclareMathOperator{\Sg}{Sg}

\DeclareMathOperator{\PPP}{P}
\DeclareMathOperator{\HHH}{H}
\DeclareMathOperator{\SSS}{S}
\DeclareMathOperator{\HSP}{HSP}
\DeclareMathOperator{\HS}{HS}
\DeclareMathOperator{\Clo}{Clo}
\DeclareMathOperator{\RelClo}{RelClo}
\DeclareMathOperator{\Pol}{Pol}
\DeclareMathOperator{\Inv}{Inv}

\DeclareMathOperator{\proj}{pr}
\DeclareMathOperator{\Var}{Var}

\DeclareMathOperator{\Power}{P}

\DeclareMathOperator{\CSP}{CSP}

\begin{document}

\title{Strong subalgebras and the Constraint Satisfaction Problem}

\author{Dmitriy Zhuk\email{zhuk.dmitriy@gmail.com}
}

\institute{Department of Mechanics and Mathematics, Lomonosov Moscow State University, Russia
}

\maketitle

\begin{abstract}
In 2007 it was conjectured that the Constraint Satisfaction Problem (CSP) over a constraint language $\Gamma$ is tractable 
if and only if $\Gamma$ is preserved by a weak near-unanimity (WNU) operation.
After many efforts and partial results,
this conjecture was independently proved by Andrei Bulatov and the author in 2017. 
In this paper we consider one of two main ingredients of my proof,
that is, strong subalgebras that allow us to reduce domains of the variables iteratively.
To explain how this idea works we show the algebraic properties of strong subalgebras and provide 
self-contained proof of two important facts about the complexity of the CSP.
First, we prove that if a constraint language is not preserved by a WNU operation then 
the corresponding CSP is NP-hard.
Second, we characterize all constraint languages that can be solved by local consistency checking.
Additionally, we characterize all idempotent algebras 
not having a WNU term of a concrete arity $n$, not having a WNU term, 
having WNU terms of all arities greater than 2.
Most of the results presented in the paper are not new, but I believe 
this paper can help to understand my approach to CSP and 
the new self-contained proof of known facts 
will be also useful.
\end{abstract}

\tableofcontents



\section{Introduction}

The \emph{Constraint Satisfaction Problem (CSP)} is the problem of deciding whether there is an assignment to a set of variables
subject to some specified constraints.
In general, this problem is NP-complete but if we restrict the constraint language it can be solved in polynomial time (\emph{tractable}).
Formally, for a set of relations $\Gamma$, called \emph{the constraint language}, by $\CSP(\Gamma)$ we denote the following decision problem:
given a formula
$$R_{1}(v_{1,1},\ldots,v_{1,n_{1}})
\wedge
\dots
\wedge
R_{s}(v_{s,1},\ldots,v_{1,n_{s}}),$$
where $R_{1},\dots,R_{s}\in \Gamma$, 
and $v_{i,j}\in \{x_{1},\dots,x_{n}\}$ for every $i,j$;
decide whether this formula is satisfiable.
In 1998 it was conjectured by Feder and Vardi that 
$\CSP(\Gamma)$ is either tractable, or NP-complete, 
and this conjecture is known as the CSP Dichotomy conjecture.
For two element domain 
this conjecture was proved by Schaefer
in 1974 who described all tractable cases \cite{Schaefer}.
In 2003 Andrei Bulatov classified all tractable 
cases for the 3-element domain.
It has been known since 1998 that 
 the complexity of $\CSP(\Gamma)$ depends only on the operations preserving the constraint language $\Gamma$,
 where a $k$-ary  operation $f$ \emph{preserves}
an $m$-ary relation $R$ if 
whenever $(a^1_1,\ldots,a^m_1),\ldots,(a^1_k,\ldots,a^m_k)$ in $R$, then also $(f(a^1_1,\ldots,a^1_k),\ldots,f(a^m_1,\ldots,a^m_k))$ in $R$.
Thus, for two element domain 
$\CSP(\Gamma)$ is tractable if and only if $\Gamma$ is preserved by a constant operation, a majority operation, an affine operation, conjunction or disjunction.
For over twenty years of intensive research 
many results of the form ``the existence of an operation with this property ensures tractability'' appeared \cite{Num4,Num20,Num22,Num23,CSPconjecture,BulatovAboutCSP}.
In 2007 it was shown that 
widely believed criteria for the CSP to be tractable are 
equivalent to the existence of a 
weak near-unanimity (WNU) operation preserving the constraint language 
$\Gamma$, where an operation $w$ is called a \emph{weak near-unanimity operation}
if 
$$w(y,x,x,\dots,x) = w(x,y,x,\dots,x) = \dots = 
w(x,\dots,x,y).$$

As a result, in 2007 the CSP Dichotomy Conjecture was formulated in the following nice form.

\begin{conj}\label{CSPDichotomyConjecture}
$\CSP(\Gamma)$ is tractable if $\Gamma$ is preserved by a WNU operation, 
$\CSP(\Gamma)$ is NP-complete otherwise.
\end{conj}

The hardness part of this conjecture follows from \cite{miklos}
and the only remaining part was to find a polynomial algorithm 
for $\CSP(\Gamma)$ if $\Gamma$ is preserved by a WNU operation.
Nevertheless, the conjecture remained open 
till 2017 when two independent proofs, by Andrei Bulatov and by the author, appeared.

The crucial idea of the author's algorithm came from the 
Rosenberg's description of maximal clones \cite{rosmax}, 
where \emph{a clone} is a set of operations closed under composition and containing all projections.
In fact, if a WNU operation $w$ preserves the constraint language $\Gamma$, then the domain $D_{x}$ of a variable $x$ can be viewed as an algebra $(D_{x};w)$.
Then we may consider the clone generated from $w$ and all constant operations.
Either this clone contains all operations, 
or $w$ and all constants belong to some maximal clone from the 
Rosenberg's description.
From this we derive that 
either the algebra $(D_{x};w)$
has a subalgebra with additional properties (we call it \emph{a strong subalgebra}), or 
there exists a congruence modulo which the algebra is affine.
It turned out that if the instance is consistent enough then 
having a solution is equivalent to having a solution where $x$ is 
from the strong subalgebra. 
Hence, if we forbid the affine case, we can iteratively reduce the domains of the variables to such subalgebras until all domains are one-element sets, which gives us a solution.

The paper is written to demonstrate my approach to the CSP.
To do this I provide self-contained proofs of key and well-known facts concerning the complexity of the CSP. 
Unlike the original proofs of these facts, I 
rely only on basic algebraic knowledge and the Galois connection between clones and relational clones. Everything else is proved in the paper.

The paper is organized as follows.
We start with well-known definitions from the universal algebra (Section 2).
Then, in Section 3 we formulate the definition and the main properties of strong subalgebras. To demonstrate how strong subalgebras work in universal algebra in Section~\ref{ExistenceWNUSection}
we characterize all idempotent algebras not having a WNU term,
all idempotent algebras not having a WNU term 
of a concrete arity $n\ge 3$, and 
all idempotent algebras having a WNU term operation of every arity greater than 2.
Note that similar results 
were proved in \cite{miklos},
but the original proof relies on tame congruence theory and requires
deep knowledge of universal algebra.

In Section 5 we apply strong subalgebras to describe the complexity of the CSP. First, we prove the hardness part of 
the CSP Dichotomy conjecture, that is, we show that 
$\CSP(\Gamma)$ is NP-hard whenever $\Gamma$ is not preserved by a 
WNU operation. Recall that this fact follows from \cite{CSPconjecture} and \cite{miklos}.
Second, we show that 
$\CSP(\Gamma)$ can be solve by local consistency checking whenever 
$\Gamma$ is preserved by a WNU operation of every arity greater than 2.
In this case affine algebras do not occur, and every domain, 
viewed as an algebra, has a strong subalgebra.
This allows us 
to reduce the domains of the variables iteratively
till the moment when all domains are one-element sets, giving us a solution. For the original proof and more information about local methods 
see \cite{kozik2016weak,bartokozikboundedwidth}.

Finally, in Section 6 we prove all the properties of 
strong subalgebras formulated in Section 3.


\section{Preliminaries}

\subsection{Algebras}
\emph{A finite algebra} is a pair $\mathbf{A}:=(A;F)$, where $A$ is a finite set, called \emph{universe},
and $F$ is a family of operations on $A$, called \emph{basic operations of $\mathbf{A}$}.
We say that $B\subseteq A$ is a \emph{subuniverse} 
if $B$ is closed under all operations of $\mathbf{A}$.
Given a subuniverse $B$ of $\mathbf{A}$ we can form
the algebra $\mathbf B$ by restricting all the operations
of $\mathbf{A}$ to the set $B$.
We say that $\mathbf{B}$ is \emph{a subalgebra} of $\mathbf{A}$
and write 
$B\le\algA$ or $\alg B\le\algA$.
A subuniverse/subalgebra is called \emph{nontrivial} if it is proper and nonempty.

An equivalence relation $\sigma$ on the universe of an algebra $\mathbf{A}$ is called \emph{a congruence}
if it is preserved by every operation of the algebra.
In this case we can define a factor algebra 
$\mathbf{A}/\sigma$. In a usual way we define a product of algebras of the same type (same arities of operations). 
For more details on factor algebras and product of algebras see~\cite{bergman2011universal}.
A congruence is called \emph{nontrivial} if it is 
not the equality relation and not $A^{2}$.

In this paper we usually consider \emph{idempotent} algebras, that is, 
algebras whose basic operations satisfy the condition
$f(x,x,\dots,x) = x$. 

An algebra is called \emph{essentially unary} 
if each of its operation has at most one non-dummy variable.


\subsection{Isomorphism and HSP}

We say that algebras 
$\algA=(A;f_{1},\dots,f_{s})$ and $\algB=(B;g_1,\dots,g_{s})$ are 
\emph{of the same type} if the operations $f_{i}$ and $g_{i}$ are of the same arity for every $i$. A mapping $\varphi:A\to B$ is
\emph{a homomorphism} if for every $i$ and 
$a_{1},\dots,a_{n_{i}}\in A$
$$\varphi(f_{i}(a_{1},\ldots,a_{n_{i}})) = 
g_{i}(\varphi(a_{1}),\ldots,\varphi(a_{n_{i}})).$$
If additionally $\varphi$ is bijective 
then it is called \emph{isomorphism} and we write 
$\algA\cong\algB$.
For a class $K$ of algebras of the same type 
by $\SSS(K)$ we denote the set of all subalgebras of algebras from $K$,
by $\PPP(K)$ we denote the set of all direct products of families of algebras of $K$.
By $\HHH(K)$ we denote the set of algebras 
$\algB$ such that 
$\algB\cong \algA/\sigma$ for some 
$\algA\in K$ and a congruence $\sigma$ on $\algA$.
For more details on the operators $H$, $S$ and $P$
see~\cite{bergman2011universal}.
In the paper we will use 
these notations only to write 
$\HS(\algA)$ and $\HSP(\algA)$ for an algebra $\algA$.
For example 
$\HS(\algA)$ is the set of all algebras 
$\alg B$ such that 
$\algB\cong \alg S/\sigma$ 
for some $\alg S\le \alg A$ and a congruence $\sigma$ on $\alg S$.

\subsection{Clones and relational clones}

A set of operations is called \emph{a clone} if it is closed under composition and contains all projections.
For a set of operations $F$ by $\Clo(F)$ we denote the clone generated by $F$.
For an algebra $\algA$ by $\Clo(\mathbf{A})$ we denote the clone generated by all basic operations of $\mathbf{A}$.
Operations from $\Clo(\mathbf{A})$ are called 
\emph{term operations} because they can be defined by a term over the basic operations of $\algA$.

By $\mathcal R_{A}$ denote the set of all relations over $A$, 
that is, $\mathcal R_{A}=\{R\subseteq A^{n}\mid n\in\mathbb N\}$.
Recall that an $m$-ary operation \emph{$f$ preserves a relation 
$R\subseteq A^{n}$} if 
for all $\alpha_{1},\dots,\alpha_{m}\in R$ 
we have 
$f(\alpha_{1},\dots,\alpha_{m})\in R$, 
where the operation $f$ is applied to tuples coordinatewise.
In this case we also say that $f$ is \emph{a polymorphism}
of $R$, and $R$ is \emph{an invariant} of $f$.
For a set of operations $F$
by $\Inv(F)$ we denote the set of relations preserved by each operation from $F$. 
Similarly, for an algebra $\algA$ by 
$\Inv(\algA)$ we denote 
the set of relations preserved by each basic operation of $\algA$. 
For a set of relations $\Gamma\subseteq \mathcal R_{A}$
by $\Pol(\Gamma)$ we denote the set of all operations preserving each relation from $\Gamma$.

A formula of the form 
$\exists y_{1}\dots\exists y_{n} \Phi$,
where $\Phi$ is a  conjunction of relations from $\Gamma$ is called 
\emph{a positive primitive formula (pp-formula) over $\Gamma$}.
If $R(x_{1},\dots,x_{n}) = \exists y_{1}\dots\exists y_{n} \Phi$, 
then we say that $R$ is \emph{pp-defined} by 
this formula, 
and $\exists y_{1}\dots\exists y_{n} \Phi$ is called 
its \emph{pp-definition}.
A set of relations $\Gamma$ closed under 
the pp-formulas 
and containing the equality and empty relations 
is called \emph{a relational clone}.

Note that $\Pol$ and $\Inv$ are mutually inverse bijective mappings between clones and relational clones,
defining a Galois correspondence 
between them
\cite{bond,geiger1968closed}.
Precisely, 
for any 
algebra $\algA$ 
we have 
$\Clo(\algA) = \Pol(\Inv(\algA))$,
and for any 
$\Gamma\subseteq \mathcal R_{A}$ 
we have 
$\RelClo(\Gamma) = \Inv(\Pol(\Gamma))$,
where 
$\RelClo(\Gamma)$ is 
the relational clone 
generated by $\Gamma$.

\subsection{Other notations}

For an integer $k$ by $[k]$ we denote the set $\{1,2,\dots,k\}$.
For a $k$-ary relation $R$ and a set of coordinates $B \subseteq [k]$, define $\proj_B(R)$ to be the $|B|$-ary relation obtained from $R$ by projecting onto $B$, or equivalently, existentially quantifying variables at positions $[k] \setminus B$. 
To simplify we sometimes write 
$\proj_{1}(R)$ and $\proj_{1,2}(R)$ instead 
of 
$\proj_{\{1\}}(R)$ and $\proj_{\{1,2\}}(R)$.

A relation 
$R\subseteq A_{1}\times\dots\times A_{n}$ is called 
subdirect if 
$\proj_{i}(R) = A_{i}$ for every $i\in[n]$.
We say that a subalgebra $\mathbf{R}$
of $\algA_{1}\times\dots\times \algA_{n}$ 
is \emph{subdirect}, we write 
$\alg R\le_{sd} \algA_{1}\times\dots\times \algA_{n}$,
if $R$ is a subdirect relation.
Note that 
$R\subseteq A^{n}$
is a subuniverse of 
$\algA^{n}$ if and only if
$R\in\Inv(\algA)$.
For $R\subseteq A^{n}$ by 
$\Sg_{\algA}(R)$ we denote the minimal 
subalgebra of $\algA^{n}$ containing $R$, that is 
the subalgebra of $\algA^{n}$ generated from $R$.



\section{Strong subalgebras}\label{StrongSubalgebrasSection}

In this section we give a definition of strong subalgebras 
and formulate all the necessary properties.

\subsection{Binary absorbing subuniverse}

We say $B$ is \emph{an absorbing subuniverse} of an algebra 
$\algA$ if 
there exists $t\in \Clo(\mathbf{A})$ such that
$t(B,B,\dots,B,A,B,\dots,B) \subseteq B$ for any position of $A$. Also in this case we say that 
\emph{$B$ absorbs $\algA$ with a term $t$}.

If the operation $t$ can be chosen binary then we say that
$B$ is \emph{a binary absorbing subuniverse} of $\mathbf A$.
To shorten sometimes we will write \emph{BA} instead of 
binary absorbing.
If $t$ can be chosen ternary the we call $B$ 
\emph{a ternary absorbing subuniverse}.
For more information about absorption and its connection with CSP see \cite{barto2017absorption}.

\subsection{Central subuniverse}
A subuniverse $C$ of $\algA$ is called \emph{central} 
if it is an absorbing subuniverse 
and 
for every $a\in A\setminus C$ 
we have 
$(a,a)\notin \Sg_{\algA}((\{a\}\times C)\cup (C\times \{a\}))$.
Note that this definition is weaker than the original definition of a center from my proof of CSP Dichotomy conjecture \cite{MyProofCSP,ZhukFOCSCSPPaper}. Nevertheless, it has all the good properties of a center. Moreover, unlike a center, 
a central subuniverse of a central subuniverse is a central subuniverse.
We will show (see Corollary~\ref{ternaryAbsorption}) that 
every central subuniverse is a ternary absorbing subuniverse.

\subsection{Projective subuniverse}

A subuniverse $B$ of $\algA$ is called \emph{a projective subuniverse} 
if 
every basic operation $f$ of $\algA$ has a coordinate $i$
such that 
$f(\underbrace{A,\dots,A,B}_{i},A,\dots,A)\subseteq B$.

The following easy characterization of projective subuniverses is proved 
in Subsection \ref{ProjectiveSubuniversesSubsection}.
This idea is known from the Post's description of clones 
on a 2-element domain, where $A = \{0,1\}$ and $B \in\{\{0\},\{1\}\}$
\cite{Post}.

\begin{lem}\label{ProjectiveSubCharacterization}
Suppose $\algA$ is a finite idempotent algebra.
Then 
$B$ is a projective subuniverse of $\algA$ if and only if 
$A^{n}\setminus (A\setminus B)^{n}\in\Inv(\algA)$ for every $n\ge 1$.
\end{lem}

\subsection{PC subuniverse}

An algebra $(A;F_{A})$ is called \emph{polynomially complete (PC)}
if the clone generated by $F_{A}$ and all constant operations is the clone of all operations on $A$
(see \cite{istinger1979characterization,lausch2000algebra}).
A subuniverse $B$ of $\algA$ is called 
\emph{a PC subuniverse} if
it is $A$, or empty, or 
it is a block of a congruence $\sigma$ on $\algA$ such that
$\algA/\sigma\cong \alg D_{1}\times\dots\times \alg D_{s}$, 
where each $\alg D_{i}$ is a polynomially complete algebra 
without a nontrivial binary absorbing subuniverse, 
or a nontrivial central subuniverse, or 
a nontrivial projective subuniverse.
To shorten we say that 
$\alg D_{i}$ is \emph{a PC algebra without BACP}.

\subsection{$p$-affine algebras}

For a prime number $p$ an idempotent finite algebra $\mathbf A$ is called \emph{$p$-affine}
if there exist operations $\oplus$ and $\ominus$ on $A$ such that $(A;\oplus,\ominus)\cong (\mathbb Z_{p}\times \dots\times 
    \mathbb Z_{p};+,-)$,
    $(x_{1}\oplus x_{2} = x_{3}\oplus x_{4})\in \Inv(\algA)$, and 
$(x\ominus y\oplus z)\in\Clo(\algA)$.
See \cite{freese1987commutator} for the definition and properties of affine algebras.

\subsection{Essential relations}
A relation $R\subseteq A^{n}$ is called \emph{$C$-essential} if
$R\cap(C^{i-1}\times A\times C^{n-i})\neq \varnothing$ for every 
$i\in[n]$
but $\rho\cap C^{n}=\varnothing$.
A relation 
$R\subseteq A_{1}\times\dots\times A_{n}$ is called 
\emph{$(C_{1},\dots,C_{n})$-essential}
if
$R\cap
(C_{1}\times\dots\times C_{i-1}\times A_{i} \times C_{i+1}
\times\dots\times C_{n})\neq\varnothing$
for every $i\in[n]$ 
but 
$R\cap
(C_{1}\times\dots\times C_{n})=\varnothing$.
In Subsection~\ref{absorbingSubuniverseSubsection} we will prove the following lemma, which was originally proved in \cite{DecidingAbsorption}.

\begin{lem}\label{NoEssential}\cite{DecidingAbsorption}
Suppose $B$ is a subuniverse of $\algA$.
Then $B$ absorbs $A$ with an operation $t$ of arity $n$ if and only if
there does not exist a $B$-essential relation $R\le \algA^{n}$.
\end{lem}

\subsection{Strong subalgebras}

A subuniverse $B$ of $\algA$ is called 
a \emph{strong subuniverse}
if $B$ is a binary absorbing subuniverse, 
a central subuniverse, or a PC subuniverse.
In this case the algebra $\alg B$ is called 
\emph{a strong subalgebra} of $\algA$.
If we need to specify the type of a strong subuniverse,
we say that 
$B$ is \emph{a (strong) subuniverse of $\algA$ of type $\mathcal T$}, 
where 
$\mathcal T=BA(t)$ if it is a binary absorbing subuniverse 
with a term operation $t$, 
$\mathcal T=C$ if 
it is a central subuniverse, 
and 
$\mathcal T=PC$ if 
it is a PC subuniverse.
We write 
$B\le_{\mathcal T} \algA$ or 
$\alg B\le_{\mathcal T} \algA$.

\subsection{Properties of strong subalgebras}

Strong subalgebras have a lot of nice properties, 
but here we formulate only the properties that are necessary to prove the main statements.
All these properties are proved in Section \ref{StrongSubuniversesSection}.

First, we can prove that every idempotent algebra has a strong subuniverse
(cases (1)-(3)), a projective subuniverse, or a $p$-affine quotient.
Unlike similar claims proved earlier 
(compare with Theorem 5.1 in \cite{MyProofCSP}),
this claim is formulated for all idempotent algebras but not just algebras having a WNU term operation (Taylor algebras).

\begin{thm}\label{ExistsStrongSubalgebraTHM}
Every finite idempotent algebra $\algA$ of size at least 2 has
\begin{enumerate}
\item[(1)] a nontrivial binary absorbing subuniverse, or
\item[(2)] a nontrivial central subuniverse, or
\item[(3)] a nontrivial PC subuniverse, or
\item[(4)] a congruence $\sigma$ such that 
$\algA/\sigma$ is $p$-affine, or
\item[(5)] a nontrivial projective subuniverse.
\end{enumerate}
\end{thm}

The next claim says that a projective subuniverse implies 
a binary absorbing subuniverse or 
an essentially unary algebra.

\begin{lem}\label{CBTNonAbsorbing}
Suppose $B$ is a nontrivial projective subuniverse of a finite idempotent algebra $\algA$, 
and $B$ is not a binary absorbing subuniverse.
Then there exists an essentially unary algebra 
$\alg U\in \HS(\algA)$ of size at least 2.
\end{lem}



The next theorem collects the main properties of 
subdirect subalgebras having a strong subalgebra on every coordinate.
Note that similar claims are known for absorbing subuniverses
(see Lemma 11 and Proposition 16 in \cite{DecidingAbsorption}).

\begin{thm}\label{CommonPropertiesThm}
Suppose 
$\alg R \le_{sd} \algA_{1}\times\dots\times \algA_{n}$,
$n\ge 2$,
$\algA_{1}, \dots,\algA_n$ are finite idempotent algebras, 
and 
$\alg B_{i}\le_{\mathcal T}\algA_{i}$ 
for every $i\in[n]$.
Then 
\begin{enumerate}
\item[(1)]
$(R\cap (B_{1}\times\dots\times B_{n}))\le_{\mathcal T} \alg R$;

\item[(2)]
if 
$\mathcal T\neq PC$ or 
$\algA_{1}$ has no nontrivial central subuniverses then 

$\proj_{1}(R\cap (B_{1}\times\dots \times B_{n}))\le_{\mathcal T}\algA_{1}$;


\item[(3)] 
if $R$ is $(B_{1},\dots,B_{n})$-essential
then 
$\mathcal T\in\{C,PC\}$  and $n=2$.
\end{enumerate}
\end{thm}

The next two statements explain 
how strong subuniverses of different types interact with each other.

\begin{lem}\label{PCBsub}
Suppose $\algA$ is a finite idempotent algebra, 
$B_1$ and $B_{2}$ are subuniverses of $\algA$ of types $\mathcal T_{1}$
and $\mathcal T_{2}$, respectively.
Then $B_{1}\cap B_{2}$ is strong subuniverse of $\algB_{2}$ of type $\mathcal T_{1}$.
\end{lem}

\begin{thm}\label{PCBint}
Suppose $\algA$ is a finite idempotent algebra, 
$B_{i}\le_{\mathcal T_{i}}\algA$
for every $i\in[n]$,
$n\ge 2$, 
$\bigcap_{i\in[n]}B_{i} = \varnothing$,
and 
$\bigcap_{i\in[n]\setminus\{j\}}B_{i} \neq  \varnothing$
for every $j\in[n]$.
Then 
one of the following conditions holds:
\begin{enumerate}
    \item[(1)] $n=2$ and $\mathcal T_{1} = \mathcal T_{2}\in\{C,PC\}$;
    \item[(2)] $\mathcal T_{1}, \dots, \mathcal T_{n}$ are binary absorbing types.
\end{enumerate}
\end{thm}

\section{Existence of WNU}\label{ExistenceWNUSection}
In this section we will demonstrate how 
strong subalgebras can be used to prove the crucial property of a WNU term operation concerning the Constraint Satisfaction Problem.
We characterize 
\begin{itemize}
    \item all idempotent algebras not having a WNU;
    \item all idempotent algebras not having a WNU of a concrete arity $n\ge 3$;
    \item all idempotent algebras having a WNU of every arity $n\ge 3$.
\end{itemize}

In Subsection \ref{AuxiliaryAlgebraicPropertiesSubsection} we prove auxiliary facts
about idempotent algebras.
In Subsection~\ref{ConstantTupleSubsection} we prove that every symmetric invariant 
relation satisfying some additional properties has a constant tuple. 
Then we derive the existence of a WNU  from this fact.
In Subsection~\ref{WNUBlockersSubsection} we show that 
the nonexistence of a WNU is equivalent to the existence of 
an invariant relation of a special form, which we call WNU-blockers and $p$-WNU-blockers.
In Subsection \ref{MainWNUTheoremsSubsection}
we formulate and 
prove the three characterizations we announced earlier.

\subsection{Auxiliary statements}\label{AuxiliaryAlgebraicPropertiesSubsection}

The following lemma is taken from \cite{KeyRelations}(see Lemma 6.4).

\begin{lem}\label{LinearWNU}
Suppose $(G;+)$ is a finite abelian group,
the relation $R\subseteq G^{4}$ is defined by
$R = \{(a_1,a_2,a_3,a_4)\mid a_1+a_2=a_3+a_4\}$,
$R$ is preserved by an idempotent WNU $w$.
Then $w(x_{1},\ldots,x_{n}) = t\cdot (x_{1}+\dots+x_{n})$
for some $t\in \{1,2,3,\ldots\}$.
\end{lem}

\begin{proof}

Denote $h(x) = w(0,0,\ldots,0,x)$.
Let us prove the equation \[w(x_{1},\ldots,x_{m},0,\ldots,0) = h(x_1)+\ldots+h(x_{m})\]
by induction on $m$.
For $m=1$ it follows from the definition.
We know that
$w\left(\begin{smallmatrix}
x_{1} & x_{2} & \dots & x_{m} & x_{m+1} & 0 & \dots & 0 \\
0     & 0     & \dots & 0     & 0       & 0 & \dots & 0 \\
x_{1} & x_{2} & \dots & x_{m} & 0       & 0 & \dots & 0 \\
0     & 0     & \dots & 0     & x_{m+1} & 0 & \dots & 0 \\
\end{smallmatrix}\right)\in R$,
which by the inductive assumption gives
\begin{align*}
w(x_{1},\ldots,x_{m},x_{m+1},0,\ldots,0) =\;\;\;\;\;&\\
w(x_{1},\ldots,x_{m},0,\ldots,0) + w(0,\dots,&0,x_{m+1},0,\dots,0) - w(0,\dots,0)=\\
w(x_{1},\ldots,x_{m},0,\ldots,0) + h(x_{m+1}&) =
h(x_1)+\ldots+h(x_{m})+h(x_{m+1}).
\end{align*}

Thus, we know that
$f(x_{1},\ldots,x_{n}) = h(x_1)+\ldots+h(x_{n})$.
Let $k$ be the maximal order of an element in the group $(G;+)$.
Since $w$ is idempotent, for every $a\in A$ we have $\underbrace{h(a)+h(a)+\ldots+h(a)}_n = a$.
Hence $k$ and $n$ are coprime,
and $h(x) = t\cdot x$ for any integer $t$ such that
$t\cdot n = 1 (\mod k)$.
\end{proof}

The remaining statements of this subsection let us go 
from an algebra from $\HSP(\algA)$ to an algebra 
from $\HS(\algA)$ keeping its property.

\begin{lem}\label{HSPtoHS}
Suppose 
$\algA$ is a finite idempotent algebra, 
$\alg B\in  \HSP(\algA)$, 
and $|B|>1$.
Then there exists 
$\alg B'\le \algB$ 
such that 
$|B'|>1$
and 
$\alg B'\in \HS(\algA)$.
\end{lem}

\begin{proof}
Assume that 
$\alg B\cong \alg S/\sigma$, where 
$\alg S\le \algA^{n}$.
We prove by induction on $n$. 
For $n=1$ we put $\alg B' = \alg B$.

Assume that for two equivalence classes $E_{1}$ and $E_{2}$ of
$\sigma$ 
the intersection 
$\proj_{1}(E_{1})\cap \proj_{1}(E_{2})$ is not empty.
Then 
choose $a\in \proj_{1}(E_{1})\cap \proj_{1}(E_{2})$ and put 
$S' = \proj_{2,\dots,n}(S\cap (\{a\}\times A^{n-1}))$, 
$\sigma' = \{(\alpha,\beta)\mid 
(a\alpha,a\beta)\in\sigma\}$.
Then 
the algebra $\alg S'/\sigma'$ has at least 2 elements
and is isomorphic 
to a subalgebra of $\alg S/\sigma\cong \alg B$. It remains to apply the inductive assumption to $\alg S'\le \algA^{n-1}$.

Assume that for any two equivalence classes $E_{1}$ and $E_{2}$ of
$\sigma$ 
the intersection 
$\proj_{1}(E_{1})\cap \proj_{1}(E_{2})$
is empty.
Then put
$S' = \proj_{1}(S)$, 
$$\sigma' = \{(\proj_{1}(\alpha),\proj_{1}(\beta))
\mid (\alpha,\beta)\in\sigma\}$$
and check that 
$\alg S'/\sigma'\cong \alg S/\sigma\cong\algB$.
Hence $\algB\in \HS(\algA)$.
\end{proof}

\begin{cons}\label{CorEssUnary}
Suppose 
$\algA$ is a finite idempotent algebra, 
$\alg B\in \HSP(\algA)$ is an essentially unary algebra 
of size at least 2.
Then there exists 
an essentially unary algebra 
$\alg B'\in \HS(\algA)$
of size at least 2.
\end{cons}

\begin{proof}
Any subalgebra of an essentially unary 
algebra is essentially unary. 
It remains to apply Lemma~\ref{HSPtoHS}.
\end{proof}

\begin{lem}\label{subOfpAffineispAffine}
Suppose $\algA$ is a $p$-affine algebra,
$\alg B\le\algA$, and $|B|>1$.
Then $\alg B$ is a $p$-affine algebra.
\end{lem}

\begin{proof}
Let $\varphi$ be an isomorphism 
from $(A;\oplus,\ominus)$ to $(\mathbb Z_{p}^{s};+,-)$.
Let 
$Z = \varphi(B)\subseteq \mathbb Z_{p}^{s}$.

Since 
$x\ominus y\oplus z\in\Clo(\algA)$, 
$x\ominus y\oplus z$ preserves $B$
and $x-y+z$ preserves $Z$.
Consider the algebra 
$\alg Z = (Z;x-y+z)$.
It is not hard to check that 
any term operation of $\alg Z$ 
can be represented as $c_{1}x_{1}+\dots +c_{t}x_{t}$,
where $c_{1},\dots,c_{t}\in\{0,1,\dots,p-1\}$
and $c_{1}+\dots+c_{t}=1$.
Choose a minimal generating set 
$Z_{0} = \{z_1,\dots,z_{\ell}\}$, that is a set such that 
$\Sg_{\alg Z}(Z_{0}) = Z$.
Then 
any element of 
$Z$ can be represented as
$c_{1}z_{1}+\dots +c_{\ell}z_{\ell}$, 
where $c_{1}+\dots+c_{\ell}=1$.
Since $Z_{0}$ is a minimal generating set, 
this representation is unique.
Hence 
$(B;x\ominus y\oplus z)\cong \alg Z \cong (\mathbb Z_{p}^{\ell-1};x-y+z)$.
Choose $b\in B$ 
such that $\varphi(b) = z_{\ell}$
and 
define 
operations $\oplus'$ and $\ominus'$ on $B$ by 
$x\oplus' y = x\ominus b\oplus y$,
$x\ominus' y = x\ominus y\oplus b$.
Then 
$(B;\oplus',\ominus')\cong
(\mathbb Z_{p}^{\ell-1};+,-)$,
$x\ominus y\oplus z =x\ominus' y\oplus' z$,
and
$$(x_1\oplus x_{2} = x_3\oplus x_4)\Leftrightarrow
(x_1\oplus' x_{2} = x_3\oplus' x_4).$$
Thus, we just take the operation and the relation 
witnessing that $\algA$ is $p$-affine, 
restrict them to $B$ and obtain 
an operation and a relation 
witnessing that $\alg B$ is $p$-affine. 
\end{proof}

\begin{cons}\label{corpaffine}
Suppose 
$\algA$ is a finite idempotent algebra, 
$\alg B\in \HSP(\algA)$ is a $p$-affine algebra.
Then there exists 
a $p$-affine algebra
$\alg B'\in \HS(\algA)$.
\end{cons}

\begin{proof}
By Lemma~\ref{HSPtoHS}
there exists 
$\alg B'\le \algB$ 
such that 
$\algB'\in\HS(\algA)$.
By Lemma~\ref{subOfpAffineispAffine},
$\alg B'$ is $p$-affine.
\end{proof}

\subsection{Constant tuple}\label{ConstantTupleSubsection}

A relation $R\subseteq A^{n}$ is called 
\emph{symmetric} 
if for any permutation 
$\sigma:[n]\to[n]$
$R(x_{1},\dots,x_{n}) = 
R(x_{\sigma(1)},\dots,x_{\sigma(n)})$.
In this subsection we prove that 
every symmetric relation with additional properties 
has a constant tuple and derive the conditions for the existence of a WNU term operation from this.
Note that very similar claims  were originally proved in \cite{miklos}
(see Section 4).

\begin{lem}\label{MainNoWNULemma}
Suppose 
$\algA$ is a finite idempotent algebra, 
$n\ge 3$ and $R\le \algA^{n}$  is a nonempty symmetric relation.
Then 
\begin{enumerate}
\item[(1)] $(b,b,\dots,b)\in R$ for some $b\in A$, or
\item[(2)] there exists an essentially unary algebra $\alg B\in \HS(\algA)$
with $|B|>1$, or 
\item[(3)] there exists a $p$-affine algebra $\alg B\in \HS(\algA)$, where $p$ divides $n$.
\end{enumerate}
\end{lem}

\begin{proof}
We prove by induction on the size of $\algA$.
If $|A|=1$, then condition (1) obviously holds.
Assume that we have 
a subuniverse $B$ of $\algA$ such that 
$R\cap B^{n}\neq\varnothing$.
In this case we may consider 
$R\cap B^{n}$ as an invariant relation on $\alg B$ 
and apply the inductive assumption.
Since a subalgebra of a subalgebra is a subalgebra, 
this completes the proof in this case.
Thus, below we assume that 
$R\cap B^{n}=\varnothing$ for any proper subuniverse $B$ of $\algA$.
We refer to this property as to 
\emph{the empty-property}.

Since $R$ is symmetric, 
we have 
$\proj_{1}(A)=\dots = \proj_{n}(A)$.
If $\proj_{1}(A)\neq A$, then
$\proj_{1}(A)$ is a proper subuniverse of $\algA$, 
which contradicts the empty-property.

Thus, we assume that $\proj_{1}(A) = A$.
Then we apply Theorem~\ref{ExistsStrongSubalgebraTHM} and consider 5 cases of this theorem.

\textbf{Cases (1)-(3): there exists a strong nontrivial subuniverse $B$ of $\algA$.}
Let $\alg B\le_{\mathcal T}\algA$ and $C = \proj_{2}(R\cap (B\times A\times\dots\times A))$.
Consider two cases:

Case A. Suppose $C\neq A$.
By Theorem~\ref{CommonPropertiesThm}(2) $C\le_{\mathcal T} \alg A$. Since $R$ is symmetric,
$R\cap (A\times C\times\dots\times C)\neq \varnothing$.
By the empty-property $R\cap C^{n}=\varnothing$. 
Then 
$R$ is $C$-essential, 
which contradicts Theorem~\ref{CommonPropertiesThm}(3).


Case B. Suppose $C=A$.
By the empty-property 
$R\cap B^{n}=\varnothing$.
Choose the minimal $k\in[n]$ 
such that 
$\proj_{[k]}(R)\cap B^{k} = \varnothing$.
Since $C=A$, we have  $\proj_{1,2}(R)\cap B^{2}\neq\varnothing$
and $k>2$.
Then
$\proj_{[k]}(R)$ is a $B$-essential relation of arity $k>2$, 
which contradicts Theorem~\ref{CommonPropertiesThm}(3).

\textbf{Case (4): there exists a congruence $\sigma$ on $\algA$ such that 
$\algA/\sigma$ is $p$-affine.}
If $p$ divides $n$ then this is the case (3) of the lemma.
Suppose $p$ does not divide $n$.
Consider an operation 
$m\in \Clo(\algA)$ such that 
$(m/\sigma)(x,y,z) = x\ominus y\oplus z$, 
where
$\oplus$ and $\ominus$ are from the definition of 
$p$-affine.
Choose $k$ such that 
$p$ divides $(k\cdot n-1)$.
Let 
$$t(x_{1},\dots,x_{kn}) = 
m(\dots (m(m(m(x_{1},x_{1},x_{2}),x_{1},x_{3}),x_{1},x_{4}),\dots,
x_{kn}).$$
Then 
$
(t/\sigma)(x_{1},\dots,x_{kn}) 
=x_{1}\oplus \dots\oplus x_{kn}.$
Choose a tuple $(a_{1},\dots,a_{n})\in R$,
for every $i$
by $\alpha_{i}$ we denote the tuple 
$(a_{i},\dots,a_{n},a_{1}\dots,a_{i-1})\in R$.
Let $B$ be the equivalence class of $\sigma$ defined by $ k\cdot(a_{1}/\sigma\oplus \dots \oplus a_{n}/\sigma)$ and
$$t(\alpha_{1},\dots,\alpha_{n},\alpha_{1},\dots,\alpha_{n},\dots,\alpha_{1},\dots,\alpha_{n}) = \beta.$$
By the definition of $t$ we have $\beta\in B^{n}$
and $\beta\in R$,
hence $R\cap B^{n}\neq\varnothing$, which
contradicts the empty-property.

\textbf{Case (5): there exists a nontrivial CBT subuniverse $B$.}
If $B$ is also a binary absorbing subuniverse, 
then it is the case (1) of Theorem~\ref{ExistsStrongSubalgebraTHM}.
Otherwise, by Lemma~\ref{CBTNonAbsorbing}
there exists essentially unary algebra $\alg B\in \HS(\algA)$.
\end{proof}

\begin{lem}\label{NoWNULemma}
Suppose 
$\algA$ is a finite idempotent algebra, 
$\algA$ does not have a WNU term 
operation of arity $n\ge 3$. Then
\begin{enumerate}
\item[(1)] there exists an essentially unary algebra $\alg B\in \HS(\algA)$
with $|B|>1$, or 
\item[(2)] there exists a $p$-affine algebra $\alg B\in \HS(\algA)$, where $p$ divides $n$.
\end{enumerate}
\end{lem}
\begin{proof}
Let 
$\alg D = \algA^{|A|^{2}}$.
Let $M$ be the matrix 
with 2 columns whose rows are all 
pairs from $A\times A$.
Let $\alpha$ and $\beta$ be the two columns 
of $M$. Note that 
$\alpha,\beta\in D$.
Let $R_{0}\subseteq D^{n}$ consist of all tuples
$(\alpha,\dots,\alpha,\beta,\alpha,\dots,\alpha)$
having exactly one $\beta$.
Let 
$R = \Sg_{\alg D}(R_{0})$.
By the definition, $R$ is symmetric.
Applying Lemma~\ref{MainNoWNULemma}
to $R$ and $\alg D$, we obtain one of the three cases.

Case 1. There exists a constant tuple 
$(\gamma,\dots,\gamma)\in R$, 
then there exists an $n$-ary term operation $t$ 
such that 
$t(\alpha,\dots,\alpha,\beta,\alpha,\dots,\alpha)=\gamma$ for any position of $\beta$.
From the definition of 
$\alpha$ and $\beta$ we conclude
that $t$ is an $n$-ary WNU term operation on $A$, which contradicts our assumption.

Case 2. There exists a nontrivial essentially unary algebra $\alg B\in \HS(\alg D)$,
that is $\alg B\in \HSP(\algA)$,
then 
by Corollary~\ref{CorEssUnary} we obtain 
an essentially unary $\alg B'\in \HS(\alg A)$.

Case 3. There exists a $p$-affine algebra $\alg B\in \HS(\alg D)$, where $p$ divides $n$.
that is $\alg B\in \HSP(\algA)$,
then 
by Corollary~\ref{corpaffine} there exists a $p$-affine algebra 
$\alg B'\in \HS(\algA)$.
\end{proof}

As it can be seen from the proof of the previous two lemmas, the main property we need is the existence of a strong subalgebra. 
Below we will prove two very similar lemmas having this property as the assumption.

\begin{lem}\label{HavingStrongGivesConstant}
Suppose 
$\algA$ is a finite idempotent algebra, 
every subalgebra 
$\alg B\le \alg A$ 
of size at least 2 
has a nontrivial strong subuniverse,
$n\ge 3$, and 
$R\le \algA^{n}$  is a nonempty symmetric relation.
Then 
there exists $b\in A$ such that $(b,b,\dots,b)\in R$.
\end{lem}

\begin{proof}
The proof repeats the proof of Lemma~\ref{MainNoWNULemma}
but here we have only cases (1)-(3) because 
we were promised to have a strong subuniverse.
\end{proof}

\begin{lem}\label{NoWNULemmaCopy}
Suppose 
$\algA$ is a finite idempotent algebra, 
every subalgebra 
$\alg B\le \alg A$ 
of size at least 2 
has a nontrivial strong subuniverse. 
Then $\algA$ has a WNU term operation of every arity $n\ge 3$.
\end{lem}

\begin{proof}
Here we repeat 
the proof of Lemma~\ref{NoWNULemma} but 
apply Lemma~\ref{HavingStrongGivesConstant} instead of 
Lemma~\ref{MainNoWNULemma} and consider just case 1.
The only thing missing is the fact that 
every subalgebra of 
$\alg D = \algA^{|A|^{2}}$ of size at least 2 has a strong subuniverse.

Suppose $\alg S\le \algA^{|A|^{2}}$.
Since $|S|>1$, we can choose $i$ such that 
$|\proj_{i}(S)|>1$. 
Let $C = \proj_{i}(S)$. 
Then $\alg C\le\alg A$ and there exists 
a nontrivial strong subalgebra
$\alg B\le \alg C$.
Put $S' = \{\alpha\in S\mid \alpha(i) \in B\}$.
By Theorem~\ref{CommonPropertiesThm}(1), 
$S'$ is a strong subuniverse of $\alg S$, which completes the proof.
\end{proof}

\subsection{WNU-blockers}\label{WNUBlockersSubsection}

A relation $R = (B_{0}\cup B_{1})^{3}\setminus (B_{0}^{3}\cup B_{1}^{3})$,
where $B_{0},B_{1}\subseteq A$, $B_{0}\neq \varnothing$, 
$B_{1}\neq\varnothing$, and $B_{0}\cap B_{1}=\varnothing$, 
is called \emph{a WNU-blocker}.
Such relations are similar to the 
Not-all-equal relation on a 2-element set, 
where $B_{0}$ means 0 and $B_{1}$ means 1.

Let us define a $p$-WNU blocker.
Suppose $S\subseteq A$, $s\ge 1$, $p$ is a prime number, $\varphi:S\to \mathbb Z_{p}^{s}$ is a surjective mapping.
Then the relation $R$ defined by 
$$
\{(a_{1},a_{2},a_3,a_{4}\}\mid 
a_{1},a_{2},a_{3},a_{4}\in S, 
\varphi(a_{1})+\varphi(a_{2}) = \varphi(a_{3})+\varphi(a_{4})\}$$
is called a \emph{$p$-WNU-blocker}.

As it follows from the following lemmas 
a WNU-blocker forbids an algebra to have a WNU term operation, 
and a $p$-WNU-blocker forbids 
an algebra to have 
a WNU of arity $n$ where $p$ divides $n$.

\begin{lem}\label{WNUblockerAvoidsWNU}
A WNU-blocker $R$ is not preserved by any 
idempotent WNU operation.
\end{lem}

\begin{proof}
Assume that an idempotent WNU operation $w$ of arity $m$ preserves 
a WNU-blocker $R = (B_{0}\cup B_{1})^{3}\setminus (B_{0}^{3}\cup B_{1}^{3})$.
Choose $b_{0}\in B_{0}$ and 
$b_{1}\in B_{1}$.
For $I\subseteq [m]$ 
by $\alpha_{I}$ we denote the tuple from $\{b_{0},b_{1}\}^{m}$ 
such that its $i$-th element is $b_{1}$ whenever $i\in I$.
Choose an inclusion minimal $I$ such that 
$w(\alpha_{I})\notin B_{0}$.
Since $w$ is idempotent, $|I|>0$.

Since $w(b_{1},\dots,b_{1}) = b_{1}$,
$w(b_{1},b_{0},\dots,b_{0}) =w(b_{0},\dots,b_{0},b_{1})$
and $w$ preserves $R$, 
we conclude that $w(b_{1},b_{0},\dots,b_{0})\in B_{0}$.
Thus, we proved that $|I|>1$.
Choose disjoint nonempty sets $I_{1}$ and $I_{2}$ such
that $I_{1}\cup I_{2} = I$.
Put $I' = [m]\setminus I$.
Since $w(\alpha_{I_{1}}),w(\alpha_{I_{2}})\in B_{0}$
and $w$ preserves $R$, we obtain 
$w(\alpha_{I'})\in B_{1}$.
Considering 
$w(\alpha_{I'})$, 
$w(\alpha_{I})$, 
and $w(b_{1},\dots,b_{1})$, 
we derive that 
$w(\alpha_{I})\in B_{0}$, 
which contradicts our assumption.
\end{proof}

\begin{lem}\label{pWNUblockerAvoids}
A $p$-WNU-blocker $R$ is not preserved by any 
idempotent WNU operation $w$ of arity $n$, where 
$p$ divides $n$.
\end{lem}

\begin{proof}
Assume that $R$ is preserved by $w$, where $R$ is defined by
$$
\{(a_{1},a_{2},a_3,a_{4}\}\mid 
a_{1},a_{2},a_{3},a_{4}\in S, 
\varphi(a_{1})+\varphi(a_{2}) = \varphi(a_{3})+\varphi(a_{4})\}.$$
Then $w$ preserves 
$S$, which can be pp-defined by $R(x,x,x,x)$.
Let $\sigma(x,y)$ be the equivalence relation on $S$ defined by
$\varphi(x) = \varphi(y)$,
which can be pp-defined by $R(x,x,x,y)$.
Then the quotient 
$w/\sigma$ preserves the relation 
$(y_{1}+ y_{2} = y_{3} +  y_{4})$ on 
$S/\sigma$.

By Lemma~\ref{LinearWNU}, 
$(w/\sigma)(x_{1},\dots,x_{n}) = t\cdot(x_{1}+ \dots + x_{n})$.
Since 
$p$ divides $n$, 
for every $a\in S$
the element 
$w(a,a,\dots,a)$ is from the equivalence class of $\sigma$ corresponding to $(0,\dots,0)$, 
which contradicts the idempotency.
\end{proof}

\begin{lem}\label{EssentiallyUnaryImpliesWNUBlocker}
Suppose 
$\algA$ is a finite idempotent algebra, 
$\alg B\in \HS(\algA)$ is an essentially unary algebra of size at least 2,
then there exists a WNU-blocker $R\in \Inv(\algA)$.
\end{lem}

\begin{proof}
Assume that 
$\alg B\cong \alg S/\sigma$ and $\alg S\le \algA$.
Let $B_{0}$ and $B_{1}$ be two equivalence classes of $\sigma$.
Then 
the relation 
$(B_{0}\cup B_{1})^{3}\setminus (B_{0}^{3}\cup B_{1}^{3})$ is an invariant of 
$\algA$.
\end{proof}

\begin{lem}\label{pAffineImpliespWNUBlocker}
Suppose 
$\algA$ is a finite idempotent algebra and
$\alg B\in \HS(\algA)$ is a $p$-affine algebra.
Then there exists a $p$-WNU-blocker $R\in \Inv(\algA)$.
\end{lem}

\begin{proof}
Assume that $\alg B\cong \alg S/\sigma$ and $\alg S\le \algA$.
Since $\alg B$ is $p$-affine,
we have 
$(B;\oplus,\ominus)\cong (\mathbb Z_{p}^{s};+,-)$.
Let 
$\varphi$ be a natural 
mapping from $S$ to $\mathbb Z_{p}^{s}$ 
Put 
$$R = \{(a_{1},a_{2},a_3,a_{4}\}\mid 
a_{1},a_{2},a_{3},a_{4}\in S, 
\varphi(a_{1})+\varphi(a_{2}) = \varphi(a_{3})+\varphi(a_{4})\}.$$
Note that $R$ is a $p$-WNU-blocker.
Since 
$S$ is a subuniverse of $\alg A$, $\sigma$ is a congruence on $\alg S$,  and $R$ restricted to $S$ is almost the relation from the definition of a $p$-affine algebra,
we have
$R\in\Inv(\algA)$. 
\end{proof}

\begin{lem}\label{WNUimpliesTWOWNU}
Suppose 
$R\in\Inv(\algA)$ is a WNU-blocker.
Then there exists 
a $2$-WNU-blocker $R'\in\Inv(\algA)$.
\end{lem}

\begin{proof}
Suppose $R = (B_{0}\cup B_{1})^{3}\setminus (B_{0}^{3}\cup B_{1}^{3})$. Put $S = B_{0}\cup B_{1}$ and
$$
R''(x_1,x_2,x_3,x_4) = 
\exists y
\exists z\;
R(x_{1},x_{2},y)\wedge
R(x_{3},x_{3},z)\wedge
R(x_{4},y,z).
$$
It follows from the definition that 
$$R'' = S^{4}\setminus
((B_{0}\times B_{0}\times B_{0}\times B_{1})\cup 
(B_{1}\times B_{1}\times B_{1}\times B_{0})).$$
Then the required 2-WNU-blocker $R'$ can be defined by 
\begin{align*}
R'(x_1,x_2,x_3,x_4) = 
R''(x_1,x_2,x_3,x_4)&\wedge
R''(x_1,x_2,x_4,x_3)\wedge\\
R''(&x_3,x_4,x_1,x_2)\wedge
R''(x_3,x_4,x_2,x_1).
\end{align*}
In fact, if
$\varphi:S\to \{0,1\}$ maps elements of $B_{0}$ and 
$B_{1}$ to 0 and 1, respectively,
then
$$R' = \{(a_{1},a_{2},a_3,a_{4}\}\mid 
a_{1},a_{2},a_{3},a_{4}\in S, 
\varphi(a_{1})+\varphi(a_{2}) = \varphi(a_{3})+\varphi(a_{4})\}.$$
\end{proof}

\subsection{Main theorems}\label{MainWNUTheoremsSubsection}

In this subsection we prove three characterizations announced earlier.

\begin{thm}\label{CharacterizationOfNWNUTHM}
Suppose 
$\algA$ is a finite idempotent algebra and $n\ge 3$.
Then the following conditions are equivalent:
\begin{enumerate}
\item[(1)] there does not exist a WNU term operation of arity $n$;
\item[(2)] there exists an essentially unary algebra $\alg B\in \HS(\algA)$ of size at least 2, or
there exists a $p$-affine algebra $\alg B\in \HS(\algA)$, where $p$ divides $n$;
\item[(3)] there exists $R\in\Inv(\algA)$
that is a WNU-blocker or a $p$-WNU-blocker, where $p$ divides $n$.
\end{enumerate}
\end{thm}

\begin{proof}
By Lemma~\ref{NoWNULemma}, 
(1) implies (2).
By Lemmas \ref{EssentiallyUnaryImpliesWNUBlocker} and 
\ref{pAffineImpliespWNUBlocker}, (2) implies (3).
By Lemmas \ref{WNUblockerAvoidsWNU} and \ref{pWNUblockerAvoids},
(3) implies (1).
\end{proof}

Note that in the next theorem the equivalence of 
the conditions 
(1) and (3) follows from \cite{miklos}, 
the equivalence of (1) and (2) is proved in \cite{cyclicterms} 
(see Theorem 4.2).

\begin{thm}\label{CharacterizationOfAWNUTHM}
For every finite idempotent algebra $\algA$ the following conditions are equivalent:
\begin{enumerate}
\item[(1)] there exits a WNU term operation;
\item[(2)]  there exits a WNU term operation of each prime arity $p>|A|$;
\item[(3)]  there does not exist an essentially unary algebra $\alg B\in \HS(\algA)$ of size at least 2;
\item[(4)]  there does not exist a WNU-blocker $R\in\Inv(\algA)$.
\end{enumerate}
\end{thm}

\begin{proof}
$(1)\Rightarrow (4)$ is by  Lemma~\ref{WNUblockerAvoidsWNU}.
$(4)\Rightarrow (3)$ is by Lemma~\ref{EssentiallyUnaryImpliesWNUBlocker}.

$(3)\Rightarrow (2)$.
Assume that (2) does not hold, 
then there exists a prime number $p>|A|$
such that $\algA$ does not have a WNU term operation of arity $p$.
By Lemma~\ref{NoWNULemma}, 
there exists a nontrivial essentially unary algebra $\alg B\in \HS(\algA)$, or
there exists a $p'$-affine algebra $\alg B\in \HS(\algA)$, where $p'$ divides $p$.
Since $p'\le |B|\le |A|<p$ and $p$ is prime, 
the second condition cannot hold.
Hence, 
there exists an essentially unary algebra $\alg B\in \HS(\algA)$ of size at least 2, contradicting (3).

$(2)\Rightarrow (1)$ is obvious.
\end{proof}

\begin{thm}\label{CharacterizationOfALLWNUTHM}
For every finite idempotent algebra $\algA$ the following conditions are equivalent:
\begin{enumerate}
\item[(1)] there exits a WNU term operation of every arity $n\ge 3$;
\item[(2)] for some $k\ge 3$ there exits a WNU term operation of every arity $n\ge k$;
\item[(3)] there does not exist a $p$-WNU-blocker $R\in\Inv(\algA)$;
\item[(4)] there does not exist an essentially unary algebra $\alg B\in \HS(\algA)$ of size at least 2,
and there does not exist a $p$-affine algebra $\alg B\in \HS(\algA)$;
\item[(5)] every subalgebra 
$\alg B\le \alg A$ 
of size at least 2 
has a nontrivial strong subuniverse. 
\end{enumerate}
\end{thm}

\begin{proof}
$(1)\Rightarrow(2)$ is trivial.
$(2)\Rightarrow(3)$
is by Lemma \ref{pWNUblockerAvoids}.

$(3)\Rightarrow(4)$. We will prove that negation 
of (4) implies the negation of (3).
Consider two cases.
Case 1. Assume that there exists an essentially unary algebra $\alg B\in \HS(\algA)$.
By Lemma~\ref{EssentiallyUnaryImpliesWNUBlocker}, 
there exists 
a WNU-blocker $R\in\Inv(\algA)$.
By Lemma~\ref{WNUimpliesTWOWNU} there exists 
a $2$-WNU-blocker $R'\in \Inv(\algA)$. 
Case 2. Assume that there exists a $p$-affine algebra $\alg B\in \HS(\algA)$. 
By Lemma~\ref{pAffineImpliespWNUBlocker}, 
there exists a $p$-WNU-blocker in $\Inv(\algA)$.

$(4)\Rightarrow(5)$. 
Consider $\algB\le \algA$.
Apply Theorem~\ref{ExistsStrongSubalgebraTHM} to $\alg B$ and consider 5 cases of this theorem.
In cases (1)-(3) we just obtain a strong subalgebra.
In case (4) we obtain a $p$-affine algebra 
$\alg B/\sigma$, which contradicts condition (4).
In case (5) Lemma~\ref{CBTNonAbsorbing}
gives us a nontrivial binary 
absorbing subuniverse, which is a strong subuniverse, 
or an essentially unary algebra $\alg U\in\HS(\algA)$, 
which contradicts condition (4).

$(5)\Rightarrow(1)$ is 
by Lemma~\ref{NoWNULemmaCopy}.
\end{proof}
\section{Constraint Satisfaction problem}

In this section we will demonstrate how 
strong subalgebras can be used to study the complexity of the  Constraint Satisfaction Problem 
for 
different constraint languages.

\subsection{CSP Dichotomy Conjecture}

Suppose we have a 
finite set of relations $\Gamma\subseteq \mathcal R_{A}$, called 
\emph{a constraint language}.
Recall that 
\emph{the Constraint Satisfaction Problem over the constraint language $\Gamma$},  denote by $\CSP(\Gamma)$, is the 
following decision problem:
given a formula
$$R_{1}(v_{1,1},\ldots,v_{1,n_{1}})
\wedge
\dots
\wedge
R_{s}(v_{s,1},\ldots,v_{1,n_{s}}),$$
where $R_{1},\dots,R_{s}\in \Gamma$, 
and $v_{i,j}\in \{x_{1},\dots,x_{n}\}$ for every $i,j$;
decide whether this formula is satisfiable.
We call each
$R_{i}(v_{i,1},\ldots,v_{i,n_{i}})$ \emph{a constraint}.

It is well known that many combinatorial problems can be expressed as $\CSP(\Gamma)$
for some constraint language $\Gamma$.
Moreover, for some sets $\Gamma$ the corresponding decision problem can be solved in polynomial time (tractable);
while for others it is NP-complete.
It was conjectured that
$\CSP(\Gamma)$ is either in P, or NP-complete \cite{FederVardi}.
In 2017, two independent proofs of these conjecture
(called the CSP Dichotomy Conjecture) appeared, 
and the characterization of the tractable constraint languages 
turned out to be very simple.

\begin{thm}\label{CSPDichotomyTHM}\cite{BulatovProofCSP,MyProofCSP,ZhukFOCSCSPPaper}
Suppose $\Gamma\subseteq R_{A}$ is a finite set of relations.
Then $\CSP(\Gamma)$ can be solved in polynomial time if there exists a WNU operation
preserving $\Gamma$;
$\CSP(\Gamma)$ is NP-complete otherwise.
\end{thm}

In this section we demonstrate how strong subalgebras
can be used to prove the 
hardness part of Theorem~\ref{CSPDichotomyTHM}, 
and to characterize all constraint languages $\Gamma$ 
such that $\CSP(\Gamma)$ can be solved by local consistency checking.

\subsection{Reduction to a core}\label{ReductionToAcoreSubsection}

First, we need to show that it is sufficient to consider only idempotent case, that is the case when $\Gamma$ contains all the relations of the form $x=a$ for $a\in A$
(we call them \emph{constant relations}).
Suppose $f$ is a unary polymorphism of $\Gamma$, 
and $f(\Gamma)$ is a constraint language with domain $f(A)$
defined by $f(\Gamma) = \{f(R)\mid R\in \Gamma\}$.
It is easy to see that an instance 
of $\CSP(\Gamma)$ is equivalent to 
the corresponding instance of $\CSP(f(\Gamma))$
where we replace each relation $R_{i}$ 
by $f(R_{i})$. So the following lemma holds.

\begin{lem}\label{ReductionToACore}\cite{jeavons1998algebraic}
Suppose $f$ is a unary polymorphism of $\Gamma$.
Then $\CSP(\Gamma)$ is polynomially equivalent to 
$\CSP(f(\Gamma))$.
\end{lem}
Thus, if there exists a unary polymorphism that is not a bijection, 
this polymorphism can be used to reduce the 
domain.

A constraint language is called \emph{a core} if every unary polymorphism of $\Gamma$ is a bijection.
It is not hard to show that if $f$ is a unary polymorphism of $\Gamma$ with minimal range,
then $f(\Gamma)$ is a core \cite{CSPconjecture}. Another important fact is that we can add all
constant relations to a core constraint language
without increasing the complexity of its $\CSP$.
First we will need an auxiliary fact.

\begin{lem}\label{QuantifierFreeDefinitionOfSigma}
Suppose $A=\{0,1,\dots,k-1\}$, $\Gamma\subseteq\mathcal R_A$
and 
$\algA = (A;\Pol(\Gamma))$.
Then $\Sg_{\algA^{k}}(
\{(0,1,\dots,k-1)\})$ has a quantifier-free pp-definition over $\Gamma$.
\end{lem}
\begin{proof}
Let us show that $\Sg_{\algA^{k}}(
\{(0,1,\dots,k-1)\})$ is defined by 
$$\sigma(z_{0},\dots,z_{k-1}) = \bigwedge_{R\in \Gamma, 
(a_1,\dots,a_{s})\in R} R(z_{a_1},z_{a_2},\dots,z_{a_s}).$$
By the definition, $(0,1,\dots,k-1)\in\sigma$. 
If $(b_0,b_1,\dots,b_{k-1})\in\sigma$, 
then the definition of $\sigma$ just says that the unary function $g(x) = b_x$ preserves every relation 
$R\in\Gamma$, and this is exactly what we need.
\end{proof}

\begin{thm}\label{addingIdempotency}\cite{CSPconjecture}
Let $\Gamma\subseteq\mathcal R_A$ be a core constraint language, and
$\Gamma' = \Gamma\cup \{x=a\mid a\in A\}$.
Then $\CSP(\Gamma')$ is polynomially reducible to $\CSP(\Gamma)$.
\end{thm}

\begin{proof}
Assume that 
$A =\{0,1,\dots,k-1\}$. 
Let $\algA = (A;\Pol(\Gamma))$
and $\sigma = \Sg_{\algA^{k}}(
\{(0,1,\dots,k-1)\})$.
By Lemma~\ref{QuantifierFreeDefinitionOfSigma}
$\sigma$ has a quantifier-free pp-definition over $\Gamma$.
Suppose we have an instance $\mathcal I'$ of 
$\CSP(\Gamma')$. 
Choose $k$ new variables 
$z_{0},\dots,z_{k-1}$ and replace every constraint 
$x=a$ by $x=z_{a}$. 
Also, add the pp-definition of the constraint 
$\sigma(z_{0},\dots,z_{k-1})$.
The obtained instance we denote by 
$\mathcal I$.
Let us show that 
$\mathcal I$ is equivalent to $\mathcal I'$.

$\mathcal I'\Rightarrow \mathcal I$.
To get a solution of $\mathcal I$ it is sufficient to take a solution of 
$\mathcal I'$ and put 
$z_{a} = a$ for every $a\in A$.

$\mathcal I\Rightarrow \mathcal I'$.
Consider a solution of $\mathcal I$.
Let $(z_{0},\dots,z_{k-1}) = 
(b_{0},\dots,b_{k-1})$ in this solution.
Since $(b_{0},\dots,b_{k-1})\in\sigma$ and $\Gamma$ is a core, 
there exists a unary bijective polymorphism 
$\varphi$ such that 
$\varphi(0,1,\dots,k-1) = 
(b_{0},b_{1},\dots,b_{k-1})$.
Composing $\varphi$ we can define a unary bijective polymorphism 
$\psi$ such that 
$\psi(b_{0},b_{1},\dots,b_{k-1}) = 
(0,1,\dots,k-1)$.
Applying $\psi$ to the solution of $\mathcal I$ 
we get a solution of $\mathcal I'$.
\end{proof}

Combining Lemma~\ref{ReductionToACore} and 
Theorem~\ref{addingIdempotency} we conclude that 
it is sufficient to consider only idempotent case to
prove the CSP Dichotomy Conjecture. 

\subsection{Hardness result}

Note that the original proof of the following hardness result 
is a combination of \cite{CSPconjecture} and \cite{miklos}.
We will derive the hardness result from Theorem \ref{CharacterizationOfAWNUTHM}, which is 
very similar to the claim proved in \cite{miklos}.

\begin{thm}\cite{CSPconjecture,miklos}
Suppose $\Gamma$ does not have a WNU polymorphism, 
then 
$\CSP(\Gamma)$ is NP-hard.
\end{thm}

\begin{proof}
Consider a unary polymorphism $f$ of $\Gamma$ 
with a minimal range.
Then 
$f(\Gamma)$ is a core constraint language.
Put $\Gamma' = f(\Gamma)\cup\{x=a\mid a\in f(A)\}$.
By Lemma~\ref{ReductionToACore} and Theorem~\ref{addingIdempotency}, 
$\CSP(\Gamma)$ is polynomially equivalent $\CSP(\Gamma')$.
Let $B = f(A)$, and $\alg B = (B;\Pol(\Gamma'))$.
If $\Gamma'$ has a WNU polymorphism $w'$, then 
$w(x_{1},\dots,x_{m}) = 
w'(f(x_{1}),\dots,f(x_{m}))$ is a WNU polymorphism 
of $\Gamma$.
Since $\Gamma$ does not have a WNU polymorphism, 
neither do $\Gamma'$.
By Theorem~\ref{CharacterizationOfAWNUTHM}, there exists 
a WNU-blocker $R\in\Inv(\alg B)$.
By the Galois connection we know that 
$R$ is pp-definable over $\Gamma'$ 
(we also need the equality and empty relations but they can always be propagated out from the pp-definition of $R$).
Let us show that 
$\CSP(\Gamma')$ is NP-hard. 
Let $\mathrm{NAE}_3$ be the ternary relation on $\{0,1\}$ containing 
all tuples except for $(0,0,0)$ and $(1,1,1)$.
Consider an instance 
$\mathcal I$ of $\CSP(\{\mathrm{NAE}_3\})$, 
which is known to be an NP-hard problem \cite{Schaefer}.
Replace each $\mathrm{NAE}_3$-relation by $R$, then replace each $R$ by
its pp-definition over $\Gamma'$ (all existentially quantified variables are the new variables of the instance). The obtained instance is equivalent to $\mathcal I$, which proves that $\CSP(\Gamma')$ is NP-hard.
\end{proof}

\subsection{Cycle-consistency}

In this subsection we 
introduce the notion of cycle-consistency (similar to singleton-arc-consistency in \cite{kozik2016weak}) , which is a type of local consistency, 
and show that if a constraint language 
has a WNU polymorphism of every arity $n\ge 3$
then every cycle-consistent instance has a solution.
In Subsection \ref{whenccworkssubsection} we will argue that 
whenever the instance can be solved by local methods, 
it can be solved by local consistency checking.
For more information about 
notions of local consistency see 
\cite{bartokozikboundedwidth,kozik2016weak}.

We will need several definitions.
Here we assume that 
every variable $x$ has its own domain $D_{x}\subseteq A$. We also require each $D_{x}$ to be pp-definable over
the constraint language $\Gamma$
to guarantee that every operation from $\Pol(\Gamma)$ 
preserves $D_{x}$.
For an instance $\mathcal I$ by 
$\Var(\mathcal I)$ we denote the set of 
all variables appearing in $\mathcal I$.

For an instance $\mathcal I$ and a set of variables 
$(y_{1},\dots,y_{t})$ by
$\mathcal I(y_{1},\dots,y_{t})$ we denote the set of tuples
$(b_{1},\dots,b_{t})$ such that $\mathcal I$ 
has a solution with 
$y_{1} = b_{1},\dots,y_{t} = b_{t}$.
Thus, this expression can be viewed as a pp-definition of a relation of the arity $t$ over $\Gamma$.

We say that $z_{1}-C_{1}-z_{2}-\dots - C_{l-1}-z_{l}$ is
\emph{a path} in a CSP instance $\mathcal I$ if $z_{i},z_{i+1}$ are in the scope of the constraint $C_{i}$ for every $i\in[l]$.
We say that \emph{a path $z_{1}-C_{1}-z_{2}-\dots -C_{l-1}-z_{l}$  connects $b$ and $c$}
if there exists $a_{i}\in D_{z_{i}}$ for every $i\in[l]$
such that
$a_{1} = b$, $a_{l} = c$, and
the projection of $C_{i}$ onto $z_{i}, z_{i+1}$
contains the tuple $(a_{i},a_{i+1})$.
A CSP instance is called \emph{1-consistent} if 
the projection of any constraint $C$ onto any of its variable 
$x$ is equal to $D_{x}$.
A CSP instance is called \emph{cycle-consistent} if
it is 1-consistent and 
for every variable $z$ and $a\in D_{z}$
any path starting and ending with $z$ in $\mathcal I$
connects $a$ and $a$.

An instance is called a \emph{tree-instance}
if there is no a path
$z_{1}-C_{1}-z_{2}-\dots -z_{l-1}-C_{l-1}-z_{l}$
such that $l\ge 3$, $z_{1} = z_{l}$, and all the constraints $C_{1},\ldots,C_{l-1}$ are different.
A instance $\mathcal J$ is called 
\emph{a covering} of an instance 
$\mathcal I$ if there 
exists a mapping 
$\psi:\Var(\mathcal J)\to \Var(\mathcal I)$
such that
for every constraint $R(x_{1},\ldots,x_{n})$ of $\mathcal J$ the expression
$R(\psi(x_{1}),\ldots,\psi(x_{n}))$ 
is a constraint of $\mathcal I$.
We say that a covering is 
a \emph{tree-covering} if it is a tree-instance.
An important property of a tree-covering 
$\Upsilon$
of a 1-consistent instance 
is that 
$\Upsilon(y) = D_{y}$ for every variable $y$, 
that is we can choose any value for the variable $y$ 
a extend this value to a solution of $\Upsilon$.

Let $\Power(A)=\{B\mid B\subseteq A\}$. 
A mapping $D^{(\top)}\colon\Var(\mathcal I)\to \Power(A)$ is called 
\emph{a reduction} of the instance $\mathcal I$ if 
$D^{(\top)}_{x}\subseteq D_{x}$ 
and $D^{(\top)}_{x}$ is pp-definable over $\Gamma$ for every $x\in\Var(\mathcal I)$. 

We usually denote reductions by $D^{(1)},D^{(j)},D^{(\top)},D^{(\bot)}$.
We say that 
$D^{(\bot)}\le D^{(\top)}$ if 
$D^{(\bot)}_{y}\subseteq D^{(\top)}_{y}$ for every $y$.
We say that a reduction $D^{(\top)}$ is \emph{empty}
if $D_{y}^{(\top)}=\varnothing$ for every variable $y$.
For an instance $\mathcal I$
by $\mathcal I^{(\top)}$ we denote the instance 
obtained from $I$ by 
\begin{enumerate}
    \item replacing the domain of each variable $x$ by
    $D^{(\top)}_{x}$, and
    \item replacing every constraint $R(x_{1},\dots,x_{n})$
    by $R'(x_{1},\dots,x_{n})$, where 
    $R' = R\cap (D^{(\top)}_{x_1}\times\dots\times D^{(\top)}_{x_n})$.
\end{enumerate}
Note that 
the constraint relations of $\mathcal I^{(\top)}$
are not longer from $\Gamma$, but they are always pp-definable over $\Gamma$.
A reduction $D^{(\top)}$ of an instance $\mathcal I$ 
is called \emph{1-consistent} if 
the instance $\mathcal I^{(\top)}$ is 1-consistent.

Note that any reduction of an instance can be naturally extended 
to a covering of the instance, 
thus we assume that any reduction is automatically defined on 
a covering.

The next Lemma has its roots in Theorem 20 from \cite{FederVardi},
where the authors proved that bounded width 1
is equivalent to tree duality. 

\begin{lem}\label{ConstraintPropagation}
Suppose $D^{(\top)}$ is a reduction of an instance $\mathcal I$ and
$D^{(\bot)}$ is a maximal 1-consistent reduction of $\mathcal I^{(\top)}$.
Then for every variable $y$ of $\mathcal I$ there exists a tree-covering  $\Upsilon_{y}$ of $\mathcal I$
such that
$\Upsilon_{y}^{(\top)}(y)$ defines $D_{y}^{(\bot)}$.
\end{lem}

\begin{proof}
The proof is based on the constraint propagation procedure.
We start with an empty instance $\Upsilon_{y}$ (empty tree-covering) for every variable $y$ of $\mathcal I$, 
we iteratively change these tree-coverings to obtain
tree-coverings defining the reduction $D^{(\bot)}$ as required.
At the beginning 
the 
reduction $D^{(\bot)}$
is defined by 
$D^{(\bot)}_{y} := \Upsilon_{y}^{(\top)}(y) = D^{(\top)}_{y}$
for every variable $y$.

If at some step the reduction $D^{(\bot)}$ is 1-consistent,
then we are done. 
Otherwise, consider a constraint $R(z_{1},\dots,z_{l})$ of $\mathcal I$ that breaks 1-consistency of $\mathcal I^{(\bot)}$, which means that 
the restriction of the variables $z_{1},\ldots,z_{l}$ to 
the sets $D^{(\bot)}_{z_{1}},\dots,D^{(\bot)}_{z_{1}}$
implies a stronger restriction of some variable $z_{i}$.
We change the tree-covering $\Upsilon_{z_{i}}$
by
$\Upsilon_{z_{i}}:= R(z_{1},\dots,z_{l})\wedge \Upsilon_{z_{1}}\wedge\dots\wedge\Upsilon_{z_{l}}$,
where we rename the variables of 
$\Upsilon_{z_{1}},\dots,\Upsilon_{z_{l}}$ so that 
they did not have common variables inside the new definition 
of $\Upsilon_{z_{i}}$,
and $\Upsilon_{z_{i}}$ stayed a tree-covering.
As a result we reduce 
the domain $D^{(\bot)}_{z_{i}}$ to 
$\Upsilon_{z_{i}}^{(\top)}(z_{i})$.

Since our instance is finite and every time we reduce some domain, this procedure will stop eventually giving us the required reduction $D^{(\bot)}$.

It remains to explain why the reduction 
$D^{(\bot)}$ is maximal.
Consider  a 1-consistent reduction 
$D^{(1)}\le D^{(\top)}$.
Since $\Upsilon_{y}$ is a tree-covering,
we have 
$\Upsilon_{y}^{(1)}(y) = D_{y}^{(1)}$.
Hence 
$D_{y}^{(1)} = \Upsilon_{y}^{(1)}(y)
\subseteq \Upsilon_{y}^{(\top)}(y)
=D_{y}^{(\bot)}$.
\end{proof}

The following theorem was originally proved in \cite{kozik2016weak}.

\begin{thm}\label{CycleConsistencyImpliesSolutionIdemp}\cite{kozik2016weak}
Suppose $\Gamma\subseteq \mathcal R_{A}$ is a constraint language containing all constant relations and having a WNU polymorphism of every arity $n\ge 3$, and 
$\mathcal I$ is a cycle-consistent instance of $\CSP(\Gamma)$.
Then $\mathcal I$ has a solution.
\end{thm}

\begin{proof}
Let $\algA = (A;\Pol(\Gamma))$.
Since every domain $D_{x}$ has a pp-definition over $\Gamma$, it is a subuniverse of the algebra $\algA$.
Hence, we may consider a subalgebra 
$\alg D_{x}\le \algA$ for every variable $x$ of $\mathcal I$.

The idea of the proof is to build a sequence of 
1-consistent reductions 
$D^{(0)}\ge D^{(1)}\ge \dots\ge D^{(t)}$ 
of the instance $\mathcal I$
such that 
\begin{enumerate}
\item[(1)] $D^{(0)}_{y} = D_{y}$ for every $y\in\Var(\mathcal I)$;
\item[(2)] $|D^{(t)}_{y}| = 1$ for every $y\in\Var(\mathcal I)$;
\item[(3)] $D^{(i)}_{y}\le_{\mathcal T_{i}}\alg D^{(i-1)}_{y}$ for every $y\in\Var(\mathcal I)$ and $i\in[t]$.
\end{enumerate}
Obviously, 
the 1-consistent reduction 
$D^{(t)}$ gives us a solution of the instance.

Assume that we already have 
reductions 
$D^{(0)}, D^{(1)},\dots,D^{(s)}$.
Let us show how to build the next reduction.

By Theorem~\ref{CharacterizationOfALLWNUTHM}
every algebra $\alg D_{x}^{(s)}$
of size at least two has a nontrivial strong subalgebra. 
If for some variable $x$ the algebra
$\alg D_{x}^{(s)}$ has a nontrivial BA or central subuniverse, 
we choose this variable.
Otherwise we choose a variable $x$ such that 
$\alg D_{x}^{(s)}$ has a nontrivial PC subuniverse.
Thus, let 
$B$ be a nontrivial strong subuniverse of $\alg D_{x}^{(s)}$ of a type $\mathcal T_{s+1}$.


Most likely, if we just reduce 
$D_{x}^{(s)}$ to $B$ then the obtained instance 
will not be 1-consistent.
Nevertheless, we can reach 1-consistency.

By $D^{(\top)}$ we denote the 
reduction that coincides with 
$D^{(s)}$ on all variables but $x$ and $D_{x}^{(\top)} = 
B$. 
By $D^{(\bot)}$ we denote a maximal 1-consistent reduction of $\mathcal I^{(\top)}$, which theoretically can be empty.
By Lemma~\ref{ConstraintPropagation} 
for every variable $y\in\Var(\mathcal I)$ 
there exists a tree-covering $\Upsilon_{y}$ of $\mathcal I$
such that $\Upsilon_{y}^{(\top)}(y)$ defines 
$D_{y}^{(\bot)}$.
We consider two cases.

Case 1. The reduction $D^{(\bot)}$ is empty,
then the instance $\Upsilon_{y}^{(\top)}$ does not have a solution.
We will show that this case cannot happen.
To simplify our notations we put $D^{(s+1)} = D^{(\top)}$. 
Let $X$ be the set of variables from $\Upsilon_{y}$ that are mapped to $x$ in the definition of the covering $\Upsilon_{y}$.
We know that $\Upsilon_{y}^{(s)}$ does not have a solution 
such that all of these variables are from $B$.
Choose a minimal set $\{x_{1},\dots,x_{\ell}\}\subseteq X$ of variables we need to restrict to $B$ so that 
$\Upsilon_{y}^{(s)}$ has no solutions.
Let $\Upsilon_{y}^{(s)}(x_{1},\dots,x_{\ell})$ define 
a relation $R$.
Then $R$ is a $B$-essential relation.
Since the instance $\mathcal I^{(s)}$ is 
1-consistent and 
$\Upsilon_{y}$ is a tree-covering of 
$\mathcal I$, the relation $R\le \alg D_{x}^{(s)}\times \dots\times \alg D_{x}^{(s)}$ is subdirect
and $\ell\ge 2$. 

By Theorem \ref{CommonPropertiesThm}(3),
$\mathcal T_{s+1}$ cannot be an absorbing type and 
$\ell=2$.
Let $\Var(\Upsilon_{y}) =
\{x_{1},x_{2},y_{1},\dots,y_{u}\}$
and $S$ be the solution set of 
$\Upsilon_{y}$, that is a relation defined by 
$\Upsilon_{y}(x_{1},x_{2},y_{1},\dots,y_{u})$.
Since $\mathcal I^{(i)}$ is 1-consistent for any $i\le s$ and 
$\Upsilon_{y}$ is a tree-covering, 
the relation $S^{(i)}$ is subdirect.

For $0\le i, j,k\le s+1$ put
$
S_{i,j,k}= S\cap (D_{x_{1}}^{(i)}\times D_{x_{2}}^{(j)}\times D_{y_{1}}^{(k)}\times\dots\times D_{y_{u}}^{(k)})
$.


Let us show by induction on $i$ that  
$S_{i,i,k+1}$ is a strong subuniverse of 
$\alg S_{i,i,k}$
for any $0\le k\le i\le s$.
For $i=k$ this follows from 
Theorem~\ref{CommonPropertiesThm}(1) and the fact that $S^{(k)}$ is subdirect.
Let us prove the induction step assuming that 
$S_{i,i,k+1}$ is a strong subuniverse of 
$S_{i,i,k}$.
Since $S^{(i)}$ is subdirect,
Theorem~\ref{CommonPropertiesThm}(1) implies that $S_{i+1,i+1,k}$ is a strong subuniverse of 
$S_{i,i,k}$.
Then by Lemma~\ref{PCBsub}, 
we obtain that $S_{i+1,i+1,k}\cap S_{i,i,k+1} = S_{i+1,i+1,k+1}$
is a strong subuniverse of $S_{i+1,i+1,k}$.
Thus, $S_{i,i,k+1}$ is a strong subuniverse of 
$S_{i,i,k}$
for any $0\le k\le i\le s$.


Since $R\cap B^{2}=\varnothing$, 
we have $S_{s+1,s+1,s} = \varnothing.$
Since 
$\mathcal I$ is cycle-consistent,
the path from $x_1$ to $x_2$ in $\Upsilon_{y}$ 
connects any $a\in D_{x}^{(s+1)}$ with $a$.
Therefore. $S_{s+1,s+1,0}\neq \varnothing$.
Choose a minimal $k<s$ 
such that
$S_{s+1,s+1,k+1} = \varnothing.$
Thus, we have 
\begin{equation}\label{ThreeIntersection}
S_{s+1,s+1,k+1} = S_{s+1,s,k}\cap S_{s,s+1,k}\cap S_{s,s,k+1}
=\varnothing
\end{equation}
We already proved that 
$S_{s,s,k+1}$ is 
a strong subuniverse of $S_{s,s,k}$.
Since $S^{(s)}$ is subdirect, 
Theorem~\ref{CommonPropertiesThm}(1) implies that
$S_{s+1,s,k}$ and $S_{s,s+1,k}$ are strong subuniverses of 
$S_{s,s,k}$.
Thus, in (\ref{ThreeIntersection}) we 
have an intersection of three
strong subuniverses 
of $S_{s,s,k}$.

Since $k$ was chosen minimal, 
$S_{s+1,s+1,k}=S_{s+1,s,k}\cap S_{s,s+1,k}$.  is not empty.
Since $S^{(s)}$ is subdirect,
$S_{s,s+1,s}\neq \varnothing$ and $S_{s+1,s,s}\neq \varnothing$.
Therefore 
$S_{s,s+1,k+1} = S_{s,s+1,k}\cap S_{s,s,k+1}\neq \varnothing$
and 
$S_{s+1,s,k+1} = S_{s+1,s,k}\cap S_{s,s,k+1}\neq \varnothing$.
Thus, the intersection of any two 
strong subuniverses in (\ref{ThreeIntersection}) 
is not empty.
Since 
$S_{s+1,s,k}$ is not a binary absorbing subuniverse, 
we get a contradiction
with Theorem~\ref{PCBint}.

Case 2. The reduction $D^{(\bot)}$ is not empty.
Put $D^{(s+1)} = D^{(\bot)}$.
Let $\Var(\Upsilon_{y}) = 
\{y,x_{1},\dots,x_{s}\}$ 
and 
$R_{y}$ be defined by 
$\Upsilon_{y}^{(s)}(y,x_{1},\dots,x_{s})$.
Since 
$\mathcal I^{(s)}$ is 1-consistent, 
the relation $R_{y}$ is subdirect.
We know that 
$D_{y}^{(s+1)} = 
\proj_{1}(R_{y}\cap (D_{y}^{(\top)}\times 
D_{x_{1}}^{(\top)}\times \dots\times D_{x_{s}}^{(\top)}))$.
By 
Theorem~\ref{CommonPropertiesThm}(2),
$D_{y}^{(s+1)}$
is a strong subuniverse of $D_{y}^{(s)}$ of type $\mathcal T_{s+1}$.
Thus, we proved that the new reduction $D^{(s+1)}$ 
satisfies condition (3), and we made our sequence of reductions longer.
\end{proof}

\subsection{An algorithm for cycle-consistency checking}

\newcommand{\Changed}{\mbox{Changed}}
\newcommand{\Solve}{\mbox{Solve}}
\newcommand{\Reduce}{\mbox{Reduce}}
\newcommand{\Output}{\mbox{Output}}
\newcommand{\CheckCycleConsistency}{\mbox{\textsc{CheckCC}}}

The cycle-consistency is a local property and it can be checked in polynomial time. 
Moreover, we can either find a cycle-consistent reduction of
the instance, or prove that it has no solutions.

\begin{algorithm}
\begin{algorithmic}[1]
\Function{\CheckCycleConsistency}{$\mathcal I$}
\State{\textbf{Input:} CSP($\Gamma$) instance $\mathcal I$}
    \For{$u,v\in\Var(\mathcal I)$} 
    \Comment{Calculate binary projections $R_{u,v}$}
        \State{$R_{u,v} := D_{u}\times D_{v}$}
        \For{$C\in\mathcal I$}        
            \State{$R_{u,v} :=R_{u,v}\cap \proj_{u,v} C$}
        \EndFor
    \EndFor
    \Repeat \Comment{Propagate constraints to reduce  $R_{u,v}$}
        \State{$\Changed:= false$}
        \For{$u,v,w\in\Var(\mathcal I)$} 
            \State{$R_{u,v}'
            :=R_{u,v}\cap (R_{u,w}\circ R_{w,v})$}
            \If{$R_{u,v} \neq R_{u,v}'$}
                \State{$R_{u,v}:=R_{u,v}'$}
                \State{$\Changed:= true$}               
            \EndIf
        \EndFor
    \Until{$\neg\Changed$} \Comment{We cannot reduce $R_{u,v}$ anymore}
    \For{$u\in\Var(\mathcal I)$} 
        \State{$D^{(\bot)}_{u}:= \proj_{1}(R_{u,u})$}
    \EndFor
    \Return{$D^{(\bot)}$}
\EndFunction
\end{algorithmic}
\end{algorithm}

We start with the function 
$\CheckCycleConsistency$ checking the cycle-consistency
of a CSP instance $\mathcal I$ (see the pseudocode). 
First, for every pair of different variables $(u,v)$ we consider the intersections of projections of all constraints onto these variables.
The corresponding relation we denote by $R_{u,v}$.
By $\proj_{u,u}(C)$ we denote the set of all pairs $(a,a)$ such that $a\in\proj_{u}(C)$. 
Then $R_{u,u} = \{(a,b)\mid a=b\in R_{u}\}$, where 
$R_{u}$ is the intersection of the projections 
of all constraints on $u$.
Then, for every $u,v,w\in\Var(\mathcal I)$
we iteratively replace
$R_{u,v}$ by $R_{u,v}' = R_{u,v}\cap (R_{u,w}\circ R_{w,v})$,
where 
$R_{u,w}\circ R_{w,v}$ is the composition of binary relations, 
that is 
$(R_{u,w}\circ R_{w,v})(x,y)  = \exists  z \; R_{u,w}(x,z)\wedge R_{w,v}(z,y).$
We repeat this procedure while we can change some $R_{u,v}$.
In the end we define and return the reduction 
$D^{(\bot)}$ satisfying the following lemma.
 
\begin{lem}\label{ProofCycleConsistencyFunction}
Suppose the function 
$\CheckCycleConsistency$ returns 
a reduction $D^{(\bot)}$ on a CSP instance $\mathcal I$.
If $D^{(\bot)}_{u} = D_{u}$ for every $u\in\Var(\mathcal I)$ 
then $\mathcal I$ is cycle-consistent.
Moreover, any solution of $\mathcal I$ is also a solution of 
$\mathcal I^{(\bot)}$.
\end{lem}
\begin{proof}
Assume that $D^{(\bot)}_{u} = D_{u}$ for every $u\in\Var(\mathcal I)$.
It follows from the line 6 of the pseudocode that 
for any variable $u\in\Var(\mathcal I)$ 
the projection of any constraint on $u$ 
equals $\proj_{1}(R_{u,u})  = D_{u}$.
Hence $\mathcal I$ is 1-consistent.
Let us prove the cycle-consistency.
Consider a path
$u_{1}-C_{1}-u_{2}-\dots-u_{\ell-1}-C_{\ell-1}-
u_{\ell}$
starting and ending 
with $u_{1}=u_{\ell}$. 
Since the projection 
of $C_{j}$ onto $u_{j},u_{j+1}$ 
contains $R_{u_{j},u_{j+1}}$ for every $j$, 
it is sufficient to prove 
that 
$$R_{u_1,u_\ell}\subseteq 
R_{u_{1},u_{2}}\circ 
R_{u_{2},u_{3}}\circ 
\dots
\circ 
R_{u_{\ell-1},u_{\ell}}.$$
Here we used the fact that the composition of binary relations is associative.
Since the function stopped, 
$R_{u,v}\subseteq R_{u,w}\circ R_{w,v}$ 
for all $u,v,w$. 
Thus, in the right-hand  side we first replace 
$R_{u_{1},u_{2}}\circ 
R_{u_{2},u_{3}}$ by $R_{u_1,u_3}$, 
then replace $R_{u_{1},u_{3}}\circ 
R_{u_{3},u_{4}}$ by $R_{u_1,u_4}$, and so on.
Finally we will get the required condition.

Let us prove that 
any solution of $\mathcal I$ is also a solution of 
$\mathcal I^{(\bot)}$.
In fact, since all the constraints 
$R_{u,v}(u,v)$ were derived from the original constraints, they should hold on any solution.
Thus, $\proj_{1}(R_{u,u})=D_{u}^{(\bot)}$ means 
that any solution of $\mathcal I$ should have 
$u\in D_{u}^{(\bot)}$, which completes the proof.
\end{proof}

\begin{algorithm}
\begin{algorithmic}[1]
\Function{\Solve}{$\mathcal I$}
\State{\textbf{Input:} CSP($\Gamma$) instance $\mathcal I$} 
    \State{$D^{(\bot)} := \CheckCycleConsistency(\mathcal I)$}
    \If{$D^{(\bot)}$ is empty}
        \Return{``No solution"}
    \EndIf
    \If{$D^{(\bot)}_{y}=D_{y}$ for every $y\in\Var(\mathcal I)$}
        \Return{\mbox{``Ok"}}
    \EndIf    
    \Return{$\Solve (\mathcal I^{(\bot)})$}
\EndFunction
\end{algorithmic}
\end{algorithm}

Thus, the function $\CheckCycleConsistency$ either 
gives us a reduction of the instance, or says that the instance is cycle consistent. After we found a reduction, we 
can apply it to the instance and check the cycle consistency again.
We can do this till the moment when we get 
a cycle-consistent reduction of the instance.
See the pseudocode of the function $\Solve$ doing this.
From Lemma~\ref{ProofCycleConsistencyFunction}
we can easily derive the following lemma.

\begin{lem}\label{ProofSolveFunction}
If the function 
$\Solve$ returns 
\mbox{``Ok"} then 
there exists a nonempty cycle-consistent reduction 
of the instance;
if it returns \mbox{``No solution"}
then the instance has no solutions.
\end{lem}

Let us check that the functions $\CheckCycleConsistency$ and $\Solve$ work in polynomial time.
In the function 
$\CheckCycleConsistency$
we go through the \textbf{repeat} loop 
at most $2^{|A|^{2}}\cdot |\Var(\mathcal I)|^{2}$ times, 
because every time we change at least one relation $R_{u,v}$, which is binary, 
and we have $|\Var(\mathcal I)|^{2}$ of them.
Thus, $\CheckCycleConsistency$ works in polynomial time.
In the function 
$\Solve$ every time we go deeper in the recursion, 
we reduce the domain of at least one variable
and we have polynomially many of them.
Thus, the function $\Solve$ is also polynomial.

\subsection{When cycle-consistency solves CSP}\label{whenccworkssubsection}

In this section we characterize all constraint languages $\Gamma$ for which the cycle-consistency 
guarantees the existence of a solution and, therefore, the function $\Solve$ solves $\CSP(\Gamma)$.
Note that the characterization of all CSPs that can be solved by local methods was obtained earlier 
in \cite{bartokozikboundedwidth,bulatov2009bounded}.
For more information on this topic see the surveys \cite{bartopolymorphisms,BulatovAboutCSP}.

We will prove a general result for the nonidempotent case, that is why we need to repeat the construction we made in Subsection \ref{ReductionToAcoreSubsection}.
Again we assume that 
$A = \{0,1,\dots,k-1\}$.
Suppose $\Gamma\subseteq\mathcal R_{A}$ is a constraint language, 
$f$ is a unary polymorphism with the minimal range. 
Then $f(\Gamma)$ is the constraint language with domain $f(A)$
defined by $f(\Gamma) = \{f(R)\mid R\in \Gamma\}$.
Let 
$\Gamma' = f(\Gamma)\cup \{x=a\mid a\in f(A)\}$.
As we know from Lemma~\ref{ReductionToACore}
and Theorem \ref{addingIdempotency}
the problems 
$\CSP(\Gamma)$, $\CSP(f(\Gamma))$, 
and $\CSP(\Gamma')$ are polynomially equivalent.
Recall how we build an instance of $\CSP(\Gamma)$
from an  instance of 
$\CSP(\Gamma')$.
Suppose we have an instance $\mathcal I'$ 
of $\CSP(\Gamma')$.
First, to get an instance $\mathcal I_{f}$ over $f(\Gamma)$ we add
a quantifier-free pp-definition of the relation 
$\sigma(z_{0},\dots,z_{k-1})$ over $f(\Gamma)$
(see the proof of Theorem \ref{addingIdempotency}
and Lemma~\ref{QuantifierFreeDefinitionOfSigma}),
and replace every constraint 
$(x=a)$ by $x=z_{a}$.
To get an instance $\mathcal I$ of $\CSP(\Gamma)$ we just replace every constraint relation 
$f(R)$ by $R$.
Thus we have a mapping, which we denote by $\Sigma$, 
that assigns an instance of $\CSP(\Gamma)$ to an instance of
$\CSP(\Gamma')$.
First, we show that 
$\Solve$ works identically on 
$\Sigma(\mathcal I')$ and $\mathcal I'$

\begin{lem}\label{CheckConsistencyWorksTheSame}
Suppose 
the function $\Solve$ returns 
\mbox{``Ok"} on $\mathcal I'$, 
then it returns \mbox{``Ok"} on $\Sigma(\mathcal I')$.
\end{lem}

\begin{proof}
Suppose $\mathcal I = \Sigma(\mathcal I')$.
Since 
$\Solve$ returns 
\mbox{``Ok"} on $\mathcal I'$, 
it finds a cycle-consistent reduction 
$D^{(\bot)}$ of $\mathcal I'$.
Assume that 
$\CheckCycleConsistency$ on 
$(\mathcal I')^{(\bot)}$ stopped with 
$R_{u,v}=S_{u,v}$ for $u,v\in\Var(\mathcal I')$.
Additionally, we assign
$S_{u} = \proj_{1}(S_{u,u})$, 
$S_{u,z_{i}} = S_{u}\times \{i\}$ 
$S_{z_{i},z_{j}} = \{i\}\times\{j\}$.
Thus, we defined $S_{u,v}$ for all 
variables 
$u,v\in\Var(\mathcal I)$.
Now, we argue that 
if we launch $\CheckCycleConsistency$ on $\mathcal I^{(\top)}$ for any reduction 
$D^{(\top)}\ge D^{(\bot)}$, then
$R_{u,v}\supseteq S_{u,v}$ for 
all variables $u,v\in\Var(\mathcal I)$.
It is obviously true at the beginning.

Since $\CheckCycleConsistency$ stopped with 
$S_{u,v}$, we have 
$S_{u,v} \subseteq  S_{u,w}\circ S_{w,v}$, 
for all $u,v,w\in\Var(\mathcal I')$.
If one of the variables is $z_{i}$ then the same condition holds. Recursively, we can show that 
the property $R_{u,v}\supseteq S_{u,v}$ keeps
when we execute $\CheckCycleConsistency$ on $\mathcal I^{(\top)}$.
Thus, we always have 
$R_{u,v}\supseteq S_{u,v}$,
$\CheckCycleConsistency$ cannot return a reduction 
smaller than $D^{(\bot)}$, 
and $\Solve$ returns
\mbox{``Ok"} on $\mathcal I$.
\end{proof}

Next, we generalize Theorem~\ref{CycleConsistencyImpliesSolutionIdemp} for the nonidempotent case. 

\begin{thm}\label{CycleConsistencyImpliesSolutionNonidemp}\cite{kozik2016weak}
Suppose $\Gamma\subseteq \mathcal R_{A}$ is a constraint language having a WNU polymorphism of every arity $n\ge 3$ and 
$\mathcal I$ is a cycle-consistent instance of $\CSP(\Gamma)$.
Then $\mathcal I$ has a solution.
\end{thm}
\begin{proof}
For every $n\ge 3$ the constraint language $\Gamma$ has a 
WNU polymorphism of arity $n$.
Then $f(w(x_{1},\dots,x_{n}))$ is a WNU 
polymorphism of $f(\Gamma)$.
To make this polymorphism idempotent
consider 
the function $g$ defined by
$g(x) = f(w(x,x,\dots,x))$, which is a bijection on $f(A)$
because $f$ has a minimal range.
Composing $g$ we can get $h$ such that 
$h(g(x)) = x$ on $f(A)$.
Therefore, 
$h(g(w(x_{1},\dots,x_{n})))$ gives us 
an $n$-ary idempotent WNU polymorphism of $f(\Gamma)$, 
which is also a polymorphism 
of $\Gamma'$.
Thus, we proved that $\Gamma'$ has an idempotent 
WNU polymorphism of any arity $n\ge 3$.

Consider the instance 
$\mathcal I_{0}$ of $\CSP(f(\Gamma))$
that is obtained from $\mathcal I$ by replacing of every constraint relation $R$ by $f(R)$.
Note that this is also an instance of $\CSP(\Gamma')$. 
Let us show that 
$\mathcal I_{0}$ is cycle-consistent.
Consider a path in $\mathcal I_{0}$
$$u_{1}-f(C_{1})-u_{2}-\dots-u_{\ell-1}-f(C_{\ell-1})-
u_{\ell}$$
starting and ending 
with $u_{1}=u_{\ell}$. 
Consider an element $b\in f(A)$.
Let $b = f(c)$ for some element $c\in A$.
Since $\mathcal I$ is cycle-consistent 
the same path in $\mathcal I$ connects 
$c$ and $c$. To connect $b$ and $b$ in $\mathcal I_{0}$ 
we apply $f$ to the assignment of each variable.
Thus, the instance 
$\mathcal I_{0}$ is cycle-consistent.
By Theorem~\ref{CycleConsistencyImpliesSolutionIdemp},
it has a solution.
Since $f(R)\subseteq R$ for every $R\in\Gamma$, it is also a solution of the original instance 
$\mathcal I$.
\end{proof}

The following lemma essentially says that 
a system of linear equations cannot be solved 
by local consistency checking and the function 
$\Solve$. Similar claim was originally proved in 
\cite{FederVardi}.

\begin{lem}\label{CheckCycleConsistencyFailExample}
Suppose a $p$-WNU-blocker $R$ is pp-definable over $\Gamma\subseteq \mathcal R_{A}$
and $\Gamma$ contains all constant relations.
Then there exists an instance 
$\mathcal I$ of $\CSP(\Gamma)$ such that 
it has no solutions but 
$\Solve$ returns \mbox{``Ok"} on $\mathcal I$.
\end{lem}
\begin{proof}
Consider the following system of linear equations in $\mathbb Z_{p}$.
\begin{equation}\label{LinearEquation}
\left\{
\begin{aligned} 
  x_1+x_2&=x_{3} +\;0\;\\
  x_{3} +\;0\;&=x_4+x_5\\
  x_{4} +\;0\;&=x_1+x_6\\
  x_5+x_6&=x_{2} +\;1\; 
\end{aligned}
\right.
\end{equation}
If we calculate the sum of all equations we will get 
$0=1$, which means that the system does not have a solution.
We can show that the function $\CheckCycleConsistency$ returns \mbox{``Ok"} on this system and this is almost what we need because 
the relation $R$ is very similar 
to $x_{1}+x_{2} = x_{3}+x_{4}$.
Choose two elements $c_0$ and $c_1$ from $A$ such that 
$(c_0,c_1,c_0,c_1)\in R$ and 
$(c_0,c_0,c_0,c_1)\notin R$.
Let
\begin{align*}\mathcal I_{0}  = 
((z_{1} = c_0)\wedge R(x_1,x_2,x_{3},z_{1}))\wedge 
((z_{2}=c_0) \wedge R(x_3,z_{2},x_{4},x_{5}))\wedge&\\
((z_{3}=c_0) \wedge R(x_4,z_{3},x_{1},x_{6}))\wedge
((z_{4}=c_1) \wedge R(x_5,x_6,x_{2},&z_{4})).
\end{align*}
This is almost an instance of $\CSP(\Gamma)$ and it has no solutions.
Since $R$ is pp-definable over $\Gamma$, 
we can write an equivalent instance 
$\mathcal I$ of $\CSP(\Gamma)$,
that is
$\mathcal I = \mathcal J_{1}\wedge\mathcal J_{2}\wedge\mathcal J_{3}\wedge\mathcal J_{4}$,
where each $\mathcal J_{i}$ is a pp-definition of the corresponding $R$ and the constraint $z_{i} = c$.
Thus, each $\mathcal J_{i}$ corresponds to the $i$-th equation of the system (\ref{LinearEquation}).
We assume that if $i\neq j$ then 
$\mathcal J_{i}$ and $\mathcal J_{j}$ have exactly one common variable and this variable is from the set $X = \{x_{1},\dots,x_{6}\}$.

It remains to show that $\Solve$ returns 
\mbox{``Ok"} on $\mathcal I$.
Let $V_{i} = \Var(\mathcal J_{i})$ for every $i$.
We write $u\sim v$ if $u$ and $v$ are from the same set $V_{i}$ ($\sim$ is not transitive).
For $u\in V_{i}$ 
we put 
$S_{u} = \mathcal J_{i}(u)$.
Note that 
if $u\in V_{i}\cap V_{j}$ then 
$\mathcal J_{i}(u)=\mathcal J_{j}(u)$.
For two variables 
$u,v$ from the same set $V_{i}$ put 
$S_{u,v} = \mathcal J_{i}(u,v)$.
Otherwise, 
if $u\in V_{i}$ and $v\in V_{j}$, 
we put 
$S_{u,v} = S_{u,x}\circ S_{x,v}$, 
where $\{x\} = V_{i}\cap V_{j}$.
Since any two variables of any equation in (\ref{LinearEquation}) can be chosen arbitrary to satisfy the equation, 
we have 
$S_{x_{i},x_{j}} = S_{x_{i}}\times S_{x_{j}}$
if $x_{i}\sim x_{j}$.
By $D^{(\bot)}$ we denote the reduction of 
$\mathcal I$ such that 
$D_{u}^{(\bot)} = S_{u}$ for every $u\in \Var(\mathcal I)$.

We will prove that 
if we execute 
$\CheckCycleConsistency$ on 
an instance 
$\mathcal I^{(\top)}$ with 
$D^{(\top)}\ge D^{(\bot)}$, then 
$R_{u,v}$ 
 will always contain 
$S_{u,v}$.
We can check that it is true when we start.
Then, it is sufficient to show that every time we 
calculate $R_{u,v}'$ in $\CheckCycleConsistency$,
it still contains $S_{u,v}$.
This will follow from the fact that 
$S_{u,v}\subseteq S_{u,w}\circ S_{w,v}$
for any $u,v,w\in\Var(\mathcal I)$.
Let us prove this considering 3 cases.

Case 1. If $u\sim v$ and $v\sim w$ then 
it follows from the definition.

Case 2. If $u\sim w\not\sim v$ then 
$$S_{u,w}\circ S_{w,v} = S_{u,w}\circ 
S_{w,x_{j}}\circ S_{x_{j},v}
\supseteq \\
S_{u,x_{j}}\circ S_{x_{j},v}=
S_{u,v}
$$

Case 3. If $u\not\sim w\not\sim v$ then
\begin{align*}
    S_{u,w}\circ S_{w,v} = S_{u,x_{i}}\circ S_{x_{i},w}\circ 
S_{w,x_{j}}\circ S_{x_{j},v}
\supseteq &\\
S_{u,x_{i}}\circ S_{x_{i},x_{j}}\circ S_{x_{j},v}
=S_{u,x_{i}}\circ 
(S_{x_{i}}\times &S_{x_{j}})\circ S_{x_{j},v}
=S_{u}\times S_{v}.
\end{align*}

Thus, we proved that 
in $\CheckCycleConsistency$ 
we always have 
$R_{u,v}\supseteq S_{u,v}$, hence 
$\CheckCycleConsistency$ 
can never return 
a reduction smaller than $D^{(\bot)}$ on $\mathcal I^{(\top)}$, and therefore 
$\Solve$ returns \mbox{``Ok"} on $\mathcal I$.
\end{proof}

Formally, when we reduce the domain of a variable 
of an instance, we may get a relation outside of the constraint language. 
That is why, even if we calculated (for example in 
$\CheckCycleConsistency$) that $x\in B\subsetneq D_{x}$, 
we cannot just reduce the domain of $x$.
To avoid this trouble, 
it is natural to assume that 
for any relation $R\in \Gamma$ and 
any set $D$ pp-definable over $\Gamma$ 
the relation 
$R\cap (A\times\dots\times A\times D\times A\times\dots \times A)$ 
is also from $\Gamma$.
A constraint language satisfying this property will be called 
\emph{unary-closed}.


The next theorem characterizes constraint languages solvable by local consistency checking and the function $\Solve$. 
For the original proof of similar claims see 
\cite{bartokozikboundedwidth,kozik2016weak}.

\begin{thm}
Suppose $\Gamma\subseteq \mathcal R$ is a unary-closed constraint language.
Then the following conditions are equivalent.
\begin{enumerate}
    \item[(1)] $\Gamma$ has a WNU polymorphism of every arity 
    $n\ge 3$;
    \item[(2)] every cycle-consistent instance of $\CSP(\Gamma)$ has a solution;
    \item[(3)] the function $\Solve$ solves $\CSP(\Gamma)$.
\end{enumerate}
\end{thm}

\begin{proof}
$(1)\Rightarrow(2)$ is by Theorem \ref{CycleConsistencyImpliesSolutionNonidemp}.

$(2)\Rightarrow(3)$. 
By Lemma \ref{ProofSolveFunction}
if the function 
$\Solve$ returns \mbox{``No solution"}, 
then the instance has no solutions, 
if $\Solve$ returns \mbox{``Ok"}, 
then there exists a 
nonempty cycle-consistent reduction of the instance.
Since $\Gamma$ is unary-closed, the reduction is still an instance of $\CSP(\Gamma)$ and it has a solution by (2).
Thus, $\Solve$ solves $\CSP(\Gamma)$.

$(3)\Rightarrow(1)$. Let us show that 
$\neg(1)\Rightarrow \neg(3)$.
Here we again use the definition of the unary polymorphism $f$, the constraint language 
$\Gamma'$ and the mapping $\Sigma$
(see the beginning of this subsection).
If $\Gamma$ has no a WNU polymorphism $w$ of some arity $m$,
then $\Gamma'$ has no a WNU polymorphism $w'$ of arity $m$
(otherwise,
we would get
$w(x_{1},\dots,x_{m}) = 
w'(f(x_{1}),\dots,f(x_{m}))$.
Applying
Theorem~\ref{CharacterizationOfALLWNUTHM}
to the algebra 
$(f(A),\Pol(\Gamma'))$
we conclude that there exists 
a $p$-WNU-blocker $R$
pp-definable over $\Gamma'$.
By Lemma
\ref{CheckCycleConsistencyFailExample} there exists 
an instance $\mathcal I'$ of $\CSP(\Gamma')$ such that 
$\Solve$ returns \mbox{``Ok"} on 
it and it does not have a solution.
Then 
$\Sigma(\mathcal I')$ also does not have a solution 
and by Lemma~\ref{CheckConsistencyWorksTheSame}
$\Solve$ returns \mbox{``Ok"}
on $\Sigma(\mathcal I')$.
\end{proof}

Thus, the cycle-consistency solves the constraint satisfaction problem only if 
the constraint language has a WNU polymorphism 
of every arity $n\ge 3$. 
Moreover, we can prove that 
any local method fails if the constraint language $\Gamma$ has no  WNU polymorphisms of some arity $n\ge 3$. 
In fact, in this case by Theorem~\ref{CharacterizationOfALLWNUTHM}
there exists a $p$-WNU-blocker pp-definable over $\Gamma$
and 
a $p$-WNU-blocker is like a linear equation 
$x_{1}+x_{2} = x_{3} +x_{4}$.
Using this equation and constant relations 
we can express a system of linear equations in $\mathbb Z_{p}$
as an instance of $\CSP(\Gamma)$,
and a system of linear equations cannot be solved locally (see \cite{FederVardi,bartokozikboundedwidth} for a formal statement).
Hence, we described all constraint languages that can be solved by local methods. 
Note that 
the cycle-consistency is not the weakest type of local consistency that guarantees a solution
in this case (see \cite{kozik2016weak} for more details).


\section{Strong subuniverses}\label{StrongSubuniversesSection}

In this section we will prove all the properties of 
strong subalgebras formulated in Section~\ref{StrongSubalgebrasSection}.
We assume that all algebras appearing in this section 
are finite and idempotent.
Below we give necessary definitions.

Suppose  $R\subseteq A_{1}\times\dots\times A_{n}$.
The relation $R$ is called 
\emph{full} if 
$R= A_{1}\times\dots\times A_{n}$.
It is called 
\emph{full-projective} if 
for any $I\subsetneq [n]$ 
the relation $\proj_{I}(R)$ is full.
We say that \emph{the $i$-th coordinate of $R$ is 
uniquely-determined}
if for all 
$a_{1},\dots,a_{i-1},a_{i+1},\dots,a_{n}$
there exists at most one $a_{i}\in A_{i}$ such that 
$(a_{1},\dots,a_{n})\in R$.
The relation $R$ is called \emph{uniquely-determined} if each of its coordinates is uniquely-determined.

We represent relations as matrices whose columns 
are tuples of the relations.
Sometimes we put a subset instead of an element in this matrix meaning that we can choose any element of this subset.
For instance, 
$\begin{pmatrix}
a& C \\B& c
\end{pmatrix}$
means the binary relation 
$(\{a\}\times B)\cup (C\times\{c\})$.

For a binary relation $R\subseteq A\times B$,
$A'\subseteq A$, and $B'\subseteq B$,
by 
$A'+R$ we denote
$\proj_{2}(R\cap (A'\times B))$,
by $B'-R$ we denote
$\proj_{1}(R\cap (A\times B'))$.
To shorten we
write $a+R$ instead of $\{a\}+R$.
For a congruence $\sigma$ on $\alg A$ and $B\subseteq A$ 
by $B/\sigma$ we denote the set
$\{b/\sigma\mid b\in B\}$.
Some other definitions will be given in subsections they are used.


\subsection{Absorbing subuniverse}\label{absorbingSubuniverseSubsection}
\begin{lem}\label{AbsImplies}\cite{DecidingAbsorption}
Suppose $R$ is defined by a pp-formula $\Phi$, that is, 
$$R(x_{1},\dots,x_{n}) = 
\exists y_1\dots\exists y_{s}
R_{1}(v_{1,1},\dots,v_{1,n_1})\wedge 
\dots\wedge
R_{k}(v_{k,1},\dots,v_{k,n_k}),$$
where 
$v_{i,j}\in\{x_{1},\dots,x_{n},y_{1},\dots,y_{s}\}$, 
each 
$R_{i}\le \algA_{i,1}\times \dots\times\algA_{i,n_{i}}$ , 
and $\Phi'$ is obtained from $\Phi$ by replacing of each 
relation $R_{i}$ by $R_{i}'\le_{BA(t)} \alg R_{i}$.
Then $\Phi'$ defines a relation $R'$ 
such that 
$R'\le_{BA(t)} \alg R$.
\end{lem}

\begin{proof}
Let the term $t$ be of arity $m$.
Without loss of generality we will show that 
$t(\beta_{1},\dots,\beta_{m-1},\beta_{m})\in R'$
whenever 
$\beta_{1},\dots,\beta_{m-1}\in R'$ and $\beta_{m}\in R$.
For $i\in[m-1]$ let $\gamma_{i}$ be an evaluation of 
$(y_{1},\dots,y_{s})$ in $\Phi'$ corresponding to $(x_{1},\dots,x_{n}) = \beta_{i}$.
Let $\gamma_{m}$ be the evaluation of 
$(y_{1},\dots,y_{s})$ in $\Phi$ corresponding to $(x_{1},\dots,x_{n}) = \beta_{m}$.
Since 
$R_{i}'\le_{BA(t)} \alg R_{i}$ for every $i\in[k]$,
$t(\gamma_{1},\dots,\gamma_{m})$ is a correct evaluation
of $(y_{1},\dots,y_{s})$ in $\Phi'$ corresponding to 
$(x_{1},\dots,x_{n}) = t(\beta_{1},\dots,\beta_{m})$, 
which confirms that $t(\beta_{1},\dots,\beta_{m})\in R'$.
Thus, $R'\le_{BA(t)} \alg R$.
\end{proof}

\begin{conslem}\label{AbsorptionQuotient}
Suppose $\theta$ is a congruence of $\algA$.
\begin{enumerate}
    \item If $B$ absorbs $\algA$, 
    then $B/\theta$ absorbs $\algA/\theta$
    with the same term.
        \item 
        If $B$ absorbs $\algA/\theta$, 
    then $\bigcup_{E\in B} E$ absorbs $\algA$
    with the same term.
\end{enumerate}
\end{conslem}

\begin{conslem}\label{AbsImpliesCons}
Suppose $R \le \algA_{1}\times\dots\times \algA_{n}$,
$\proj_1 (R) = A_{1}$,
$B_{i}\le_{BA(t)}\algA_{i}$ for every $i\in[n]$,
and 
$B = \proj_{1}(R\cap (B_{1}\times\dots \times B_{n}))$.
Then $B\le_{BA(t)}\algA_{1}$.
\end{conslem}

\begin{proof}
It is not hard to see that the sets $B$ and 
$A_{1}$ can be defined by
the following pp-formulas
\begin{align*}(x_1\in C) =& \exists 
x_{2}\dots\exists x_{n}\; 
\left[(x_{1}\in C_{1})\wedge 
\dots\wedge 
(x_{n}\in C_{n})\wedge 
R(x_{1},\ldots,x_{n})\right],\\
(x_1\in A_1)=& \exists 
x_{2}\dots\exists x_{n}\; 
\left[(x_{1}\in A_{1})\wedge 
\dots\wedge 
(x_{n}\in A_{n})\wedge 
R(x_{1},\ldots,x_{n})\right].
\end{align*}
It remains to apply Lemma~\ref{AbsImplies}.
\end{proof}

\begin{conslem}\label{RIntersectionBACons}
Suppose $R \le \algA_{1}\times\dots\times \algA_{n}$
and $B_{i}\le_{BA(t)}\algA_{i}$ for every $i\in[n]$.
Then $R\cap (B_{1}\times\dots\times B_{n})\le_{BA(t)}\alg R$.
\end{conslem}




\begin{lem}\label{AbsorptionReduceArity}
Suppose $B\le \alg A$,
$R\le (\alg A^{k})^{n}$ is an $n$-ary $B^{k}$-essential relation.
Then there exists 
a $B$-essential relation $R'\le \algA^{n}$.
\end{lem}

\begin{proof}
Put 
$B_{i} = B^{i}\times A^{k-i}$.
Consider a tuple $(i_1,\dots,i_{n})$ 
with the minimal sum $i_1+\dots+i_{n}$
such that 
$R\cap (B_{i_1}\times\dots\times B_{i_n})=\varnothing.$
Since $R$ is $B^{k}$-essential, 
$i_{j}\ge 1$ for every $i$.
Then 
an $n$-ary $B$-essential relation $R'$ can be defined by
(here $\alpha(i)$ is the $i$-th element of the tuple $\alpha$)
$$
\{(\alpha_{1}(i_1),\dots,\alpha_{n}(i_n))\mid 
\exists (\alpha_1,\dots,\alpha_{n})\in R\cap (B_{i_1-1}\times\dots\times B_{i_n-1})\},$$
which can be viewed as a pp-definition and therefore $R'\le \algA^{n}$.
\end{proof}

\begin{LEMNoEssential}\cite{DecidingAbsorption}
Suppose $B$ is a subuniverse of $\algA$.
Then $B$ absorbs $A$ with an operation $t$ of arity $n$ if and only if
there does not exist a $B$-essential relation $R\le \algA^{n}$.
\end{LEMNoEssential}

\begin{proof}
$\Rightarrow$. Assume that such a $B$-essential relation $R$ exists. 
Consider $n$ tuples witnessing that $R$ is $B$-essential, 
that is, 
$\alpha_{i}\in R\cap (B^{i-1}\times A\times B^{n-i})$
for $i\in[n]$.
Since $B$ absorbs $\alg A$ with a term $t$, 
$t(\alpha_{1},\dots,\alpha_{n})\in R\cap B^{n}$, which contradicts 
the fact that $R$ is $B$-essential.

$\Leftarrow.$ 
Let $M$ be the matrix whose rows are all tuples from $A^{n}$ having exactly one element outside of 
$B$. Moreover, we assume that the matrix starts with the rows whose first element outside of $B$,
then we have the rows whose second element outside of $B$, and so on.
Finally, there are rows whose last element outside of $B$.
Let $k:=(|A|-|B|)|B|^{n-1}$, then the matrix has $n\cdot k$ rows.
Let 
$R =\Sg_{\algA}(M)$.

Assume that $R\cap B^{kn}\neq\varnothing$, 
then there exists a term operation such that 
$t(M)\in B^{kn}$. By the definition of the matrix, $B$ absorbs $\algA$ 
with a term operation $t$, which is what we need.

Assume that $R\cap B^{kn}=\varnothing$.
Note that $R$ can be viewed as an $n$-ary relation on the set $A^{k}$.
In the first column of the matrix $M$ only the first $k$ elements are not from $B$, 
in the second column only $(k+1)$-th to $2k$-th elements are not from $B$, and so on.
Thus, columns of the matrix $M$ witness that $R\le (\alg A^{k})^{n}$ is 
a $B^{k}$-essential relation.
Lemma~\ref{AbsorptionReduceArity} implies that there exists a $B$-essential relation $R'\le \alg A^{n}$.
Contradiction.

\end{proof}
\subsection{Central subuniverse}

\begin{lem}\label{CenterAddingB}
Suppose 
$\alg C\le_{C} \alg A$, 
then 
$\alg C\times \alg B\le_{C}
\alg A\times\alg B$.
\end{lem}
\begin{proof}
 By Lemma~\ref{AbsImplies}, 
$C\times B$ is an absorbing subuniverse of 
$A\times B$. 
Thus, it is sufficient to show that 
for any $a\in A\setminus C$ and 
$b\in B$
we have 
$$\begin{pmatrix}
a\\b\\a\\b
\end{pmatrix}
\notin
\Sg_{\alg{A}\times\alg{B}}
\begin{pmatrix}
a& C \\b& B \\C & a\\B& b
\end{pmatrix}.$$
This follows from the fact that 
$\begin{pmatrix}
a\\a
\end{pmatrix}
\notin
\Sg_{\algA}
\begin{pmatrix}
a& C\\C & a
\end{pmatrix}.$
\end{proof}

\begin{lem}\label{CenterIntersection}
Suppose 
$\alg C_{1}\le_{C}\alg A$ 
and 
$\alg C_{2}\le_{C}\alg A$.
Then 
$C_{1}\cap C_{2} \le_{C}\alg A$.
\end{lem}
\begin{proof}
By Lemma~\ref{AbsImplies}, 
$C_1\cap C_2$ is an absorbing subuniverse of 
$A$. 
Let us show the second condition of a central subuniverse.
Suppose 
$a\in A\setminus (C_{1}\cap C_2)$,
then $a\notin C_{i}$ for some $i\in\{1,2\}$
and
$$\begin{pmatrix}
a\\a
\end{pmatrix}
\notin
\Sg_{\algA}
\begin{pmatrix}
a& C_i\\C_i & a
\end{pmatrix}
\supseteq 
\Sg_{\algA}
\begin{pmatrix}
a& C_1\cap C_2 \\
C_1\cap C_2 & a 
\end{pmatrix}.$$
\end{proof}

\begin{lem}\label{CenterIntersectionA}
Suppose 
$\alg C\le_{C}\alg A$,
$\alg B\le\alg A$.
Then 
$C\cap  B\le_{C} \alg B$.
\end{lem}
\begin{proof}
It follows from Lemma~\ref{AbsImplies}
that $C\cap B$ is an absorbing subuniverse.
The remaining part follows from the fact that for any
$a\in B\setminus C$ we have 
$$\begin{pmatrix}
a\\a
\end{pmatrix}
\notin
\Sg_{\algA}
\begin{pmatrix}
a& C \\C & a
\end{pmatrix}
\supseteq 
\Sg_{\alg B}
\begin{pmatrix}
a& C\cap B\\C\cap B & a
\end{pmatrix}
.$$%
\end{proof}

\begin{lem}\label{CenterIntersectionInIntersection}
Suppose 
$\alg C_i\le_{C} \alg A_i\le \algA$
for $i\in[k]$.
Then 
$(C_1\cap\dots \cap C_k)\le_{C}
(\alg A_1\cap\dots \cap \alg A_k)$.
\end{lem}
\begin{proof}
Since 
$(A_{1}\cap\dots\cap A_{k})\le\algA$,
by Lemma~\ref{CenterIntersectionA},
for every $i\in[k]$ we have
$$(\alg A_{1}\cap\dots\cap \alg A_{i-1}\cap \alg C_{i}
\cap \alg A_{i+1}\cap\dots\cap  \alg A_{k})
\le_{C}
(\alg A_{1}\cap\dots\cap \alg A_{k}).$$
Considering their intersection and 
applying Lemma~\ref{CenterIntersection} we complete the proof. 
\end{proof}

\begin{lem}\label{NoEEForCenter}
Suppose
$C\le_{C} \alg \algA$, 
$E\le \algA$, $E\cap C=\varnothing$,
and $e\in E$ is chosen so that 
$\Sg_{\algA}(C\cup\{e\})$ is inclusion minimal.
Then 
$\Sg_{\algA}
\begin{pmatrix}
e& C \\C & e 
\end{pmatrix}\cap E^{2}=\varnothing$.
\end{lem}

\begin{proof}
Assume the converse.
Let 
$R = \Sg_{\algA}
\begin{pmatrix}
e& C \\C & e 
\end{pmatrix}$,
$(e_1,e_2)\in R\cap E^{2}$.
Let 
$E_{1} = E+R$, then 
$C\cup \{e_2\}\subseteq E_{1}$ and by the 
minimality of
$\Sg_{\algA}(C\cup \{e\})$ 
we have $e\in E_{1}$, hence 
$(e,e')\in R$ for some $e'\in E$.
Let $E_{2} = e + R$.
We know that 
$C\cup \{e'\} \subseteq E_{2}$.
Then by the minimality of 
$\Sg_{\algA}(C\cup \{e\})$ 
we have $e\in E_{2}.$
Thus, $(e,e)\in R$, which contradicts the definition of a central subuniverse.
\end{proof}

\begin{lem}\label{CenterQuotient}
Suppose 
$\alg C\le_{C} \alg A$,
$\sigma$ is a congruence on $\alg A$.
Then 
$C/\sigma\le_{C} \alg A/\sigma$.
\end{lem}
\begin{proof}
By Corollary~\ref{AbsorptionQuotient},
$C/\sigma$ is an absorbing 
subuniverse of $\alg A/\sigma$.
Assume that 
$C/\sigma$ is not a central subuniverse.
Then there exists an equivalence class $E$
of $\sigma$
such that 
$E\cap C= \varnothing$ and 
$\begin{pmatrix}
E\\E 
\end{pmatrix}
\in 
\Sg_{\algA/\sigma}
\begin{pmatrix}
E& C/\sigma \\C/\sigma & E 
\end{pmatrix}$.
Choose 
an element $e\in E$ such that 
$\Sg_{\algA} (\{e\}\cup C)$ is 
minimal by inclusion.
By the definition of 
$E$ we have 
$\Sg_{\algA}
\begin{pmatrix}
e& C\\C & e 
\end{pmatrix}
\cap 
\begin{pmatrix}
E\\E
\end{pmatrix}\neq \varnothing$,
which contradicts Lemma~\ref{NoEEForCenter}.
\end{proof}

\begin{cons}\label{CenterProjection}
Suppose 
$\alg C\le_{C}\alg R\le 
\alg A_1\times\dots\times \alg A_{k}$,
$I\subseteq [k]$.
Then 
$(\proj_{I}C)\le_{C}(\proj_{I}\alg R)$.
\end{cons}
\begin{proof}
It is sufficient to consider 
a natural congruence 
$\sigma$ on $\alg A_1\times\dots\times \alg A_{k}$ defined by 
$(\alpha,\beta)\in\sigma\Leftrightarrow
\proj_{I}(\alpha) = \proj_{I}(\beta)$ and apply Lemma 
\ref{CenterQuotient}
\end{proof}

Now we are ready to prove the first main property of a central subuniverse, 
that is 
if we replace every relation in a pp-definition 
by its central subuniverse then we define a central subuniverse of the originally defined relation. 

\begin{thm}\label{CenterImplies}
Suppose $R$ is defined by a pp-formula $\Phi$, that is, 
$$R(x_{1},\dots,x_{n}) = 
\exists y_1\dots\exists y_{s}
R_{1}(v_{1,1},\dots,v_{1,n_1})\wedge 
\dots\wedge
R_{k}(v_{t,1},\dots,v_{k,n_k}),$$
where 
$v_{i,j}\in\{x_{1},\dots,x_{n},y_{1},\dots,y_{s}\}$, 
each 
$R_{i}\le \algA_{i,1}\times \dots\times\algA_{i,n_{i}}$ , 
and $\Phi'$ is obtained from $\Phi$ by replacing of each 
relation $R_{i}$ by $R_{i}'\le_{C} \alg R_{i}$.
Then $\Phi'$ defines a relation $R'$ 
such that 
$R'\le_{C} \alg R$.
\end{thm}

\begin{proof}
First, we want each relation $R_{i}$ and $R_{i}'$ to depend on all the variables
$x_{1},\dots,x_{n},y_{1},\dots,y_{s}$.
To achieve this
using Lemma~\ref{CenterAddingB} we add dummy variables 
to all relations.
By Lemma~\ref{CenterIntersectionInIntersection}
$R_{1}'\cap \dots\cap R_{k}'$ is a central subuniverse of 
$\alg R_{1}\cap \dots\cap \alg R_{k}$.
Applying existential quantifiers is equivalent to 
taking a projection. Thus, the final step follows from Corollary~\ref{CenterProjection}.
\end{proof}

\begin{conslem}\label{CenterQuotientTwo}
Suppose $B\le_{C}\alg A/\sigma$. 
Then 
$\bigcup_{E\in B} E\le_{C}\algA$.
\end{conslem}

\begin{conslem}\label{CenterImpliesCons}
Suppose $R \le \algA_{1}\times\dots\times \algA_{n}$,
$\proj_1 (R) = A_{1}$,
$B_{i}\le_{C}\algA_{i}$ for every $i\in[n]$,
and 
$B = \proj_{1}(R\cap (B_{1}\times\dots \times B_{n}))$.
Then $B\le_{C}\algA_{1}$.
\end{conslem}

\begin{conslem}\label{RIntersectionCenterCons}
Suppose $R \le \algA_{1}\times\dots\times \algA_{n}$
and $B_{i}\le_{C}\algA_{i}$ for every $i\in[n]$.
Then $R\cap (B_{1}\times\dots\times B_{n})\le_{C}\alg R$.
\end{conslem}



\begin{lem}\label{CenterLessThanThreeHelp}
Suppose
$n\ge 3$, 
$\alg C_{i}\le_{C} \alg A_{i}$ for $i\in[n]$,
and $R\le\alg A_{1}\times \dots \times \alg A_{n}$
is a $(C_{1},\dots,C_{n})$-essential 
relation.
Then there exists 
$(C_{1},\dots,C_{n-1},C_{1},\dots,$ $C_{n-1})$-essential
relation 
$R'\le \alg A_{1}\times \dots \times \alg A_{n-1}\times\alg A_{1}\times \dots \times \alg A_{n-1}$.
\end{lem}
\begin{proof}
Put $E = \proj_{n}(R\cap (C_{1}\times \dots\times C_{n-1}\times A_{n})).$
Choose $e\in E$ such that 
$\Sg_{\algA}(\{e\}\cup C_{n})$ is inclusion minimal 
among all choices of $e$.
Let $\sigma = \Sg_{\algA}
\begin{pmatrix}
e& C_n \\C_n & e 
\end{pmatrix}$.
By Lemma~\ref{NoEEForCenter}, $\sigma\cap E^{2}=\varnothing$.
Put
\begin{align*}R'(y_1,\ldots,y_{n-1},y_{1}',\ldots,y_{n-1}') =& \\
\exists z\exists z' \;R(y_1,\ldots,y_{n-1},z)&\wedge
R(y_{1}',\ldots,y_{n-1}',z')\wedge\sigma(z,z').
\end{align*}
Let us show that the relation $R'$ 
is $(C_{1},\dots,C_{n-1},C_{1},\dots,C_{n-1})$-essential.
Since $\sigma\cap E^{2}=\varnothing$,
$(C_{1}\times\dots\times C_{n-1}\times
C_{1}\times\dots\times C_{n-1})
\cap R' =\varnothing$.

Since 
the relation 
$R$ is 
$(C_{1},\dots,C_{n})$-essential, 
for any $i\in[n-1]$ there exists 
a tuple 
$(a_{1},\ldots,a_{n})\in R$ 
such that 
only its $i$-th element 
is not from the corresponding set 
of $(C_{1},\dots,C_{n})$.
Since $e\in E$,
there exist $c_{1},\ldots,c_{n-1}$ 
such that 
$(c_{1},\ldots,c_{n-1},e)\in 
R\cap (C_{1}\times \dots\times C_{n-1}\times A_{n})$.
Then if put $z= a_{n}$ and $z' = e$
we derive that $(a_{1},\dots,a_{n-1},c_{1},\ldots,c_{n-1})\in R'$.
Thus, for any 
$i\in[n-1]$
we build a tuple from $R'$ such that 
only its $i$-th element 
is not from the corresponding set 
of $(C_{1},\dots,C_{n-1},C_{1},\dots,C_{n-1})$.
In the same way we can build such a tuple
for each
$i\in\{n,n+1,\ldots,2n-2\}$.
\end{proof}

\begin{lem}\label{CenterLessThanThree}
Suppose
$n\ge 3$
and 
$\alg C_{i}\le_{C} \alg A_{i}$ for $i\in[n]$.
Then there does not exist a 
$(C_{1},\dots,C_{n})$-essential 
relation
$R\le\alg A_{1}\times \dots \times \alg A_{n}$.
\end{lem}

\begin{proof}
Assume that 
a
$(C_{1},\dots,C_{n})$-essential 
relation
$R\le\alg A_{1}\times \dots \times \alg A_{n}$ 
exists.
Note that from any 
$(C_{1},\dots,C_{n})$-essential relation 
$S\le \alg A_{1}\times \dots \times \alg A_{n}$ we can get 
$(C_{1},\dots,C_{n-1})$-essential relation by 
$$S'= \proj_{[n-1]}(R\cap (A_{1}\times \dots\times A_{n-1}\times C_{n})).$$
First, we derive a
$(C_{1},C_{2},C_{3})$-essential relation from $R$, 
then by Lemma~\ref{CenterLessThanThree}
we derive 
a $(C_{1},C_{2},C_{1},C_{2})$-essential relation, 
then a $(C_{1},C_{1},C_{2})$-essential relation
and again by Lemma~\ref{CenterLessThanThree}
a $(C_{1},C_{1},C_{1},C_1)$-essential relation.

Then, using Lemma~\ref{CenterLessThanThree} we can obtain 
a $C_{1}$-essential relation of any arity, which,
by Lemma~\ref{NoEssential}, contradicts 
the fact that 
$C_1$ absorbs $\algA_{1}$.
\end{proof}

\begin{conslem}\label{ternaryAbsorption}
Suppose $\alg C\le_{C}\alg A$. Then
$C$ is a ternary absorbing subuniverse of~$\alg A$.
\end{conslem}
\begin{proof}
By Lemma~\ref{CenterLessThanThree} there does not exists 
a $C$-essential relation $R\le \algA^{3}$.
Then by Lemma~\ref{NoEssential}, 
$C$ absorbs $\algA$ with a ternary term operation.
\end{proof}

\subsection{Projective subuniverses}\label{ProjectiveSubuniversesSubsection}


\begin{LEMProjectiveSubCharacterization}
Suppose $\algA$ is a finite idempotent algebra.
Then 
$B$ is a projective subuniverse of $\algA$ if and only if 
$A^{n}\setminus (A\setminus B)^{n}\in\Inv(\algA)$ for every $n\ge 1$.
\end{LEMProjectiveSubCharacterization}

\begin{proof}

$\Rightarrow$. Let us show that 
every $m$-ary operation $f$ of $\algA$ preserves the relation
$R_{n} =A^{n}\setminus (A\setminus B)^{n}$.
Assume that $f$ returns an element of $B$ whenever the $i$-th coordinate 
is from $B$. 
Consider tuples
$\alpha_1,\dots,\alpha_m\in R_{n}$.
Since the tuple $\alpha_{i}$ should contain an element of $B$, 
$f(\alpha_1,\dots,\alpha_m)$ also contains an element of $B$.
Hence, $f(\alpha_1,\dots,\alpha_m)\in R_{n}$.

$\Leftarrow$.
Assume the converse. Consider an $n$-ary operation $f$ of the algebra $\algA$
contradicting the fact that 
$\algA$ is projective.
Then for every coordinate $i\in[n]$
there exists a tuple 
$(b_{1}^{i},\dots,b_{n}^{i})$ such that 
$b_{i}^{i}\in B$ and $f(b_{1}^{i},\dots,b_{n}^{i})\notin B$.
Put 
$\beta_{j} = (b_{j}^{1},\dots,b_{j}^{n})$ for every $j\in[n]$.
Since each $\beta_{j}$ contains an element $b_{j}^{j}\in B$, 
each $\beta_{j}\in R_{n}$.
Since $f$ preserves $R_{n}$, 
$f(\beta_{1},\dots,\beta_{n})\in R$.
This contradicts our assumption that 
$f(b_{1}^{i},\dots,b_{n}^{i})\notin B$ for every $i\in[n]$.
\end{proof}

\begin{lem}\label{findBlocker}
Suppose $\alg A$ is a finite idempotent algebra,
$\Sg_{\alg A}(A^{n}\setminus (A\setminus B)^{n}) \neq A^{n}$ for every $n$,
where $\varnothing \neq B\subsetneq A$.
Then there exists a nontrivial projective subuniverse $C$ of $\algA$
such that $B\subseteq C$.
\end{lem}
\begin{proof}
Let $C\subsetneq A$ be a maximal set containing $B$ such that
$\Sg_{\alg A} (A^{n}\setminus (A\setminus C)^{n}) \neq A^{n}$ for every $n$.
Put $\sigma_n = \Sg_{\alg A}( A^{n}\setminus (A\setminus C)^{n})$, let us show that
$\sigma_{n}= A^{n}\setminus (A\setminus C)^{n}$ for every $n$.
Assume the opposite.
Then there exists
$(a_1,\ldots,a_{n})\in \sigma_n \cap (A\setminus C)^{n}$.
Let $m\in\{0,1,\ldots,n-1\}$
be the maximal number such that
$\{a_1\}\times\dots\times\{a_{m}\}\times A^{s} \not\subseteq \sigma_{m+s}$ for every $s\ge 0$.
Then for some $s'$
we have $\{a_1\}\times\dots\times\{a_{m+1}\}\times A^{s'}\subseteq \sigma_{m+s'+1}$.
Since the algebra is idempotent,
$\{a_1\}\times\dots\times\{a_{m+1}\}\times A^{s}\subseteq \sigma_{m+s+1}$
for every $s\ge s'$.
Put $$\delta_{s+1}(x_1,\ldots,x_{s+1})= \sigma_{m+s+1}(a_1,\ldots,a_{m},x_1,\ldots,x_{s+1}).$$
By the definition of
$m$ we know that $\delta_{s+1}\neq A^{s+1}$.
Since $\sigma_{m+s+1}$ is symmetric,
$\delta_{s+1}$ contains all tuples with $a_{m+1}$
and all tuples with an element from $C$.
Put
$C' = C\cup \{a_{m+1}\}$.
Therefore,
$\Sg_{\alg A}(A^{n}\setminus (A\setminus C')^{n}) \subseteq
\delta_{n}\neq A^{n}$ for every $n\ge s'+1$.
Since $\alg A$ is idempotent, 
$\Sg_{\alg A}(A^{n}\setminus (A\setminus C')^{n})=A^n$
for $n<s'+1$ would imply that 
$\Sg_{\alg A}(A^{s'+1}\setminus (A\setminus C')^{s'+1})=A^{s'+1}$.
Therefore, 
$\Sg_{\alg A}(A^{n}\setminus (A\setminus C')^{n})\neq A^n$
for every $n\ge 1$, which contradicts our assumption
about the maximality of $C$.
Hence, $A^{n}\setminus (A\setminus C)^{n}\in\Inv(\algA)$
for every $n\ge 1$ and by Lemma~\ref{ProjectiveSubCharacterization}
$C$ is a nontrivial projective subuniverse containing $B$.
\end{proof}


\begin{lem}\label{NontrivialProjectiveFromPower}
Suppose $\algA$ is a finite 
idempotent algebra,
$R$ is a notrivial 
projective subuniverse of $\algA^{n}$.
Then there exists a nontrivial projective 
subuniverse of $\algA$.
\end{lem}

\begin{proof}
We prove by induction on the arity of $R$. For $n=1$ it is trivial.
Assume that  
$\proj_{1}(R)\neq A_{1}$. We know that every basic operation $f$
of $\algA^{n}$ 
has a coordinate that preserves $R$. Hence, the same coordinate
preserves $\proj_{1}(R)$, and $\proj_{1}(R)$ is a projective subuniverse of $A_{1}$.

Otherwise, we choose any element $a\in A_{1}$
such that $R$ does not contain all tuples starting with $a$.
Then we consider 
$R' = \{(a_2,\ldots,a_{n})\mid (a,a_2,\ldots,a_n)\in R\}$.
Since every operation preserves $a$, 
$R'$ is a nontrivial projective subuniverse of $\algA_{2}\times\dots \times \algA_{n}$. 
It remains to apply the inductive assumption.
\end{proof}

\begin{lem}\label{ProjectiveQuotient}
Suppose $B$ is a projective subuniverse of $\alg A/\sigma$. 
Then 
$\bigcup_{E\in B} E$ is a projective subuniverse of $\algA$.
\end{lem}
\begin{proof}
Suppose an operation $f/\sigma$ returns an element of 
$B$ whenever its $i$-th coordinate is from $B$.
By the definition of a quotient, 
$f$ returns an element of $\bigcup_{E\in B} E$
whenever the $i$-th coordinate of $f$ is from $\bigcup_{E\in B} E$.
\end{proof}

\subsection{Central relation}

$\sigma\subseteq A^{2}$ is 
called \emph{reflexive} if $(a,a)\in \sigma$ for every $a\in A$.
A subdirect relation $R\subseteq A\times B$ is called \emph{central}
if 
$\{a\mid \{a\}\times B\subseteq R\}$ 
is nonempty and not equal to $A$.
This kind of relations came from 
the Rosenberg Classification of maximal clones \cite{rosmax}.
Such relations gave the name to central subuniverses 
as the set $\{a\mid \{a\}\times B\subseteq R\}$ can be called \emph{a center},
and we know from the following theorem that
a center is always a central subuniverse if the algebra $\alg B$ 
avoids nontrivial binary absorbing and projective subuniverses.

\begin{thm}\label{CenterImpliesTHM}
Suppose 
$\alg R\le_{sd} \alg A\times \alg B$,
$C = \{c\in A\mid \forall b\in B\colon (c,b)\in R\}$.
Then one of the following conditions holds:
\begin{enumerate}
    \item $C$ is a central subuniverse of $\alg A$;
    \item $\alg B$ has a nontrivial binary absorbing subuniverse;
    \item $\alg B$ has a nontrivial 
    projective subuniverse.
\end{enumerate}
\end{thm}
\begin{proof}
Consider $a\in A\setminus C$.
Assume that 
$\begin{pmatrix} a\\a\end{pmatrix}
\in\Sg_{\algA}
\begin{pmatrix} a & C\\
C & a 
\end{pmatrix}$.
By $a^{+}$ we denote 
the set $\{b\in B\mid(a,b)\in R\}$.
Then there exists a function 
$f\in \Clo(\alg A)$ of arity $n$ such that 
$$f(a,\dots,a,c_{1},\dots,c_{n-s})
=f(c_{1}',\dots,c_{s}',a,\dots,a) = a$$
for some 
$c_{1},\dots,c_{n-s},c_{1}',\dots,c_{s}'\in C$.
Since $f$ preserves $R$, 
we derive
for $h(x,y) = f(\underbrace{x,\dots,x}_{s},y,\dots,y)$ that 
$h(a^{+},A)\subseteq a^{+}$
and 
$h(A,a^{+})\subseteq a^{+}$.
Thus, we obtained a nontrivial binary absorbing subuniverse on $\alg B$, which corresponds to the second case.

Assume that 
$\begin{pmatrix} a\\a\end{pmatrix}
\notin\Sg_{\algA}
\begin{pmatrix} a & C\\
C & a 
\end{pmatrix}$
for every $a\in A\setminus C$.
If $C$ is an absorbing subuniverse, 
then $C$ is a central subuniverse, which is the first case.

Assume that 
$C$ is not an absorbing subuniverse.
By Lemma~\ref{NoEssential},
for every $n\ge 1$ 
there exists 
a $C$-essential subalgebra $\alg R_{n}\le \alg A^{n}$,
which means that 
for any $n$ and any $i\in\{1,2,\dots,n\}$
there exists $a_{i}^{n}\in A\setminus C$ 
such that 
$R_{n}\cap (C^{i-1}\times \{a_{i}^{n}\}\times C^{n-i})\neq \varnothing$.
Since 
$A$ is finite, there exists an element 
$a\in A\setminus C$ 
which is the most popular element 
among 
$\{a_{1}^{n},\ldots,a_{n}^{n}\}$ for infinitely many 
$n$.
For each $n$ with $a$ being the most popular
we do the following.
We restrict each variable 
of $R_{n}$ 
satisfying $a_{i}^{n}\neq a$ 
to $C$
and consider the projection 
of 
$R_{n}$ onto 
the remaining variables.
As a result we get a relation $R_{n}'$ of arity
$k_{n}\ge n/|A|$ such that 
$R_{n}'\cap C^{k_{n}}  = \varnothing$ and 
$R_{n}'\cap (C^{i}\times\{a\}\times C^{k_{n}-i})\neq\varnothing$ 
for every $i$.
Let $\alg D = \alg B^{|B|}$ and $\rho\le \alg A\times \alg D$ be the binary relation
defined by
$$\{(c,(d_{1},\ldots,d_{|B|}))\mid \forall i\colon (c,d_{i})\in R\}.$$
For each $R_{n}'$ 
let the relation 
$\Omega_{n}(y_{1},\dots,y_{k_{n}})$ be defined by 
$$
\exists x_{1}\dots\exists x_{k_{n}}\;
R_n'(x_1,\dots,x_{k_{n}})
\wedge 
\rho(x_{1},y_{1})\wedge\dots\wedge 
\rho(x_{k},y_{k_{n}}).
$$
Since $R$ is subdirect,
we can choose $a'$ such that 
$(a,a')\in R$.
Let 
$B = \{b_1,\dots,b_{|B|}\}$.
Since 
$R_{n}'\cap C^{k_{n}} = \varnothing$,
we have 
$$((b_1,\dots,b_{|B|}),\dots,(b_1,\dots,b_{|B|}))\notin \Omega_{n}.$$
Thus
$\Omega_{n}\neq D^{k_{n}}$ 
and 
$D^{k_{n}}\setminus (D\setminus\{(a',\dots,a')\})^{k_{n}}
\subseteq \Omega_{n}$. 
Since 
we consider idempotent case and $k_{n}$ can be 
as large as we need, 
we obtain that 
$\Sg_{\alg D}(D^{i}\setminus (D\setminus\{(a',\dots,a')\})^{i}
)\neq D^{i}$ for every $i$.
By Lemma~\ref{findBlocker} 
we obtain 
a nontrivial projective subuniverse of $\alg D$.
Then, by Lemma~\ref{NontrivialProjectiveFromPower},
there exists a nontrivial projective subuniverse 
on $\alg B$, which is the third case of our claim.
\end{proof}

\begin{lem}\label{CoolSubdirectImplies}
Suppose 
$R\lneq\algA_{1}\times\dots\times \algA_{n}$ is 
full-projective
and the $i$-th coordinate of $R$ is not uniquely-determined.
Then 
\begin{enumerate}
    \item there exist $l\in[n]$ and a central relation $\rho\le \algA_{l}\times \algA_{i}$, or
    \item there exists a nontrivial equivalence relation $\rho\le \algA_{i}^{2}$.
\end{enumerate}
\end{lem}

\begin{proof}
Without loss of generality we assume that $i=1$.
Choose an inclusion-maximal reflexive symmetric relation $S\lneq\algA_{1}^{2}$ 
such that its first coordinate is not uniquely-determined.
If such $S$ doesn't exist then put $S=R$.
Let $S$ be of arity $m$. Note that 
in either case 
$S$ is full-projective.

For every $k\in\{2,3,\dots,|A_{1}|\}$ we define
$\sigma_{k}(x_{1},\dots,x_{k})$ by
$$ 
\exists y_2 \dots \exists y_{m}\;
S(x_{1},y_{2},\dots,y_{m})\wedge 
\dots\wedge S(x_{k},y_{2},\dots,y_{m}).$$
Consider several cases:

Case 1. $\sigma_{|A_{1}|}$ is full. 
This means that there exists 
$(b_{2},\dots,b_{m})$ such that 
$(a,b_{2},\dots,b_{m})\in R$
for all $a\in A_{1}$.
Let $l\ge 2$ be the minimal number such that 
$S'(x_{1},\dots,x_{l}) = S(x_{1},\dots,x_{l},b_{l+1},\dots,b_{m})$ is not full.
Choose $(a_{1},\dots,a_{l})\notin S'$.
Since $S$ is full-projective, the required central relation with a center containing $b_{l}$ can be defined by 
$\rho(x,y) = S'(y,a_{2},\dots,a_{l-1},x)$.

Case 2. $\sigma_{2}$ is not full.
Since the first coordinate of $S$ is not uniquily determined, 
$\sigma_{2}$ is not uniquely-determined.
Note that $\sigma_{2}$ is a reflexive symmetric relation. 
Hence, unless $S$ is also binary reflexive symmetric relation, we get a contradiction with the choice of $S$.
Suppose $S$ is binary reflexive symmetric.
If $S$ is transitive, then $S$ is an equivalence relation, which gives us case 2.
If $S$ is not transitive, then $S\subsetneq \sigma_{2}$, giving a contradiction with the choice of $S$.

Case 3. $\sigma_{2}$ is full, $\sigma_{|A_1|}$ is not full.
Let $l$ be the minimal number such that 
$\sigma_{l}$ is not full
and $(a_{1},\dots,a_{l})\notin\sigma_{l}$.
Note that $\sigma_{l}$ is symmetric and contains all tuples with repetitive elements.
Then the required central relation 
with a center containing $\{a_{1},\dots,a_{l-2}\}$ 
can be defined by 
$\rho(x,y) = \sigma_{l}(a_{1},\dots,a_{l-2},x,y)$. 
\end{proof}

\subsection{PC subuniverses}
We start with a well-known characterization of the relations preserved by all operations (see Theorem 2.9.3 from \cite{lau}).
\begin{lem}\label{PolIsFull}
Suppose
$R\subseteq A^{n}$ is preserved by every operation on $A$.
Then $R$ can be represented as a conjunction of
binary relations of the form $x_i = x_j$. 
\end{lem}
\begin{proof}
We prove by induction on $n$.
For $n=1$ it is clear.
If $R$ is full then we are done.
Otherwise, consider $R$ as a matrix whose columns are tuples of the relation.
Since $\Pol(R)$ contains all operations 
and there exists $\alpha\notin R$,
two rows of the matrix should be equal (otherwise we could find
an operation giving $\alpha$ on the matrix).  
Thus, for some $i,j\in[n]$ 
the relation 
$\proj_{i,j}(R)$ is the equality relation.
Hence, 
$\Pol(\proj_{[n]\setminus\{i\}}(R))=\Pol(R)$ 
and the claim follows from the inductive assumption applied to 
$\proj_{[n]\setminus\{i\}}(R)$.
\end{proof}

\begin{lem}\label{ReflexivePCRelations}
Suppose
$A$ is a PC algebra and 
$R\le \algA^{n}$ contains all 
the constant tuples $(a,\dots,a)$. 
Then $R$ can be represented as a conjunction of
binary relations of the form $x_i = x_j$.\end{lem}
\begin{proof}
All constant operations preserve $R$, 
and together with the constant operations 
the algebra $\algA$ generates all operations on the set $A$.
Hence, $R$ is preserved by all operations on $A$,
and by Lemma~\ref{PolIsFull}
it can be represented
as a conjunction of
binary relations of the form $x_i = x_j$.
\end{proof}

\begin{lem}\label{PCRelationsLem}
Suppose $R\le_{sd} \algA_1\times\dots \times \algA_{n}$,
$\algA_{i}$ is a PC algebra without BACP for every $i\in\{2,\ldots,n\}$,
$\algA_1$ has no a nontrivial central subuniverse.
Then 
$R(x_{1},\dots,x_{n})=
\delta_{1}(x_{i_1},x_{j_1})\wedge \dots\wedge
\delta_{s}(x_{i_s},x_{j_s})$,
where for every $\ell\in[s]$
the first variable of $\delta_{\ell}$
is uniquely-determined whenever $i_{\ell}\neq 1$
and the second variable of $\delta_{\ell}$
is uniquely-determined whenever $j_{\ell}\neq 1$.
\end{lem}

This lemma 
says that 
the relation $R$ can be represented by 
constraints 
from the first coordinate to an $i$-th coordinate such that 
the $i$-th coordinate is uniquely-determined by the first 
(also we can define the corresponding PC congruence on the first coordinate 
using this relation)
and by bijective binary constraints between pairs of coordinates
other than first.
Also, it says that in a subdirect product of PC algebras without BACP (even $\algA_{1}$ is a PC algebra) we can choose some essential coordinates which can have any value, each other coordinate is uniquely-determined by exactly one of them (in a bijective way).

\begin{proof}
We prove this lemma by induction on the 
arity of $R$. 
Assume that 
$\proj_{I}(R)$ is not full 
for some $I\subsetneq [n]$.
Then, by the inductive assumption, one of its coordinates 
is uniquely-determined by another coordinate.
Hence some coordinate $i$ of $R$ is uniquely-determined 
by some coordinate $j$.
Let $R' =\proj_{[n]\setminus\{i\}}(R)$
and $\sigma = \proj_{i,j}(R)$.
Then 
$$R(x_{1},\dots,x_{n}) = 
R'(x_1,\dots,x_{i-1},x_{i+1},\dots,x_{n})
\wedge 
\sigma(x_{i},x_{j}).$$
By the inductive assumption, 
$R'$ and $\sigma$ can be represented as
conjunctions of proper binary relations.
Then $R$ has the required representation.

It remains to consider the case when 
$R$ is full-projective.
If some coordinate $i\in\{2,\dots,n\}$ is not uniquely-determined 
then Lemma \ref{CoolSubdirectImplies} implies  the 
existence of a nontrivial congruence on $\algA_{i}$, 
or a central relation $\rho\le\algA_{l}\times \algA_{i}$.
The first option contradicts Lemma~\ref{ReflexivePCRelations}, 
the second option by Theorem~\ref{CenterImpliesTHM}
implies the existence of a nontrivial central subuniverse on $\algA_{l}$,
or the existence of a BA/projective subuniverse on $\algA_{i}$, 
which contradicts the statement.
Thus, 
each coordinate $i\in\{2,\dots,n\}$ of $R$
is uniquely-determined.

If $R$ is binary, then we are done.

Otherwise, consider the relation 
$\zeta$ defined by
\begin{align*}\zeta(z_{1},z_{2},z_{3},z_{4}) =
\exists x_{1} \exists x_{2}\dots\exists x_{n-1} \exists x_{1}' \exists x_{2}'&\\
R(x_{1},x_{2},x_{3},\ldots,x_{n-1},z_{1})\wedge&
R(x_{1},x_{2}',x_{3},\ldots,x_{n-1},z_{2})\wedge\\
R(x_{1}',x_{2},x_{3},\ldots,x_{n-1},z_{3})\wedge&
R(x_{1}',x_{2}',x_{3},\ldots,x_{n-1},z_{4}).
\end{align*}
Since 
$R$ is full-projective,
any projection of 
$\zeta$ onto 3 variables is a full relation
and
$\zeta$ contains all constant tuples (put $x_1'=x_{1}$ and 
$x_{2}'=x_{2}$).
Then
Lemma~\ref{ReflexivePCRelations}
implies that 
$\zeta$ is a full relation.
Choose $a\neq b$ and consider 
the evaluations of the variables 
corresponding to 
$(a,a,a,b)\in\zeta$.
Since $z_{1}=z_{2}=a$ and the second variable of $R$ is uniquely-determined,
we have $x_{2} = x_{2}'$.
Since the last variable is uniquely-determined, 
we get $z_{3}=z_{4}$, that is $a=b$. 
Contradiction.
\end{proof}


\begin{conslem}\label{PCProperties}
Suppose $\sigma_{1},\ldots,\sigma_{k}$ 
are all congruences on $\algA$
such that $\algA_{i} := \algA/\sigma_{i}$ is a PC algebra without BACP,
put $\sigma = \sigma_{1}\cap \dots\cap \sigma_{k}$
and define  
$\psi:A\to A_{1}\times \dots\times A_{k}$
by $\psi(a) = (a/\sigma_{1},\dots,a/\sigma_{k})$. Then 
\begin{enumerate}
    \item[(1)] $\psi$ is surjective, hence 
    $\algA/\sigma\cong \algA_{1}\times\dots\times \algA_{k}$;
    \item[(2)] the PC subuniverses are the sets of the form 
    $\psi^{-1}(S)$, where $S\subseteq A_{1}\times \dots\times A_{k}$
    is a relation definable by unary constraints of the form 
    $x_{j} = a_{j}$;

\end{enumerate}
\end{conslem}

\begin{proof}

(1). Consider 
the image 
$\psi(A)$, which is a subdirect subuniverse 
of $\algA_{1}\times\dots\times \algA_{k}$.
By Lemma~\ref{PCRelationsLem},
this relation can be represented as a conjunction of binary relations whose one coordinate uniquely determines another (in a bijective way).
This means that congruences $\sigma_{i}$ corresponding to 
these coordinates should be equal, which contradicts the definition.
Then $\psi(A)$ is a full relation and $\psi$ is surjective.

(2). 
Suppose $B$ is a PC subuniverse.
If $B$ is empty, then the claim is trivial.
Otherwise, $B$ is a block of a congruence 
$\delta$ such that 
$\algA/\delta\cong \alg D_{1}\times\dots\times \alg D_{s}$,
where each $\alg D_{i}$ is a PC algebra without BACP.
Then for each $i\in[s]$ there exists a congruence 
$\delta_{i}\supseteq \delta$ such that 
$\algA/\delta_{i}\cong \alg D_{i}$
and $\delta_{1}\cap\dots\cap\delta_{s} = \delta$.
Then $\delta_{1},\dots,\delta_{s}$ are among 
$\sigma_{1},\dots,\sigma_{k}$, 
and the PC subuniverse $B$ is defined by fixing 
elements from the corresponding algebras among 
$\algA_{1},\dots,\algA_{k}$.
\end{proof}

\begin{conslem}\label{IntersectionOfPCisPC}
The intersection of two PC subuniverses of $\algA$ is a PC subuniverse of $\algA$.
\end{conslem}    

\begin{proof}
The statement follows from characterization of all PC subuniverses of 
$\algA$ in Corollary~\ref{PCProperties}(2).
\end{proof}

\begin{conslem}\label{EmptyPCIntersection}
Suppose $B_{i}\le_{PC} \algA$ for $i\in[n]$ and 
$B_{1}\cap\dots\cap B_{n}=\varnothing$.
Then there exist $i,j\in[n]$ such that $B_{i}\cap B_{j}=\varnothing$.
\end{conslem}    

\begin{proof}
Using claim (2) of Corollary \ref{PCProperties}, we derive that
each $B_{i}=\psi^{-1}(S_{i})$, 
where $S_{i}$ is definable by unary constraints of the form $x_{j} = a_{j}$.
Since $B_{1}\cap\dots\cap B_{n}=\varnothing$, 
two of these unary constraints contradict each other. 
Considering $S_{i}$ and $S_{j}$ giving these two constraints, 
we obtain $B_{i}$ and $B_{j}$ such that $B_{i}\cap B_{j} = \varnothing$.
\end{proof}

\begin{lem}\label{PCrestrictionOfRel}
Suppose $\alg R\le_{sd} \alg A_{1}\times\dots\times\algA_{n}$, 
$B_{i}\le_{PC} \algA_{i}$ for every $i\in[n]$. 
Then 
$(R\cap (B_{1}\times\dots\times B_{n}))\le_{PC}\alg R$.
\end{lem}

\begin{proof}
If some $B_{i}$ is empty, then the claim is trivial.
Otherwise, each $B_{i}$ is a block of a congruence $\sigma_{i}$ on $\algA_{i}$ (if $B_{i}=A_{i}$ then this congruence is $A_{i}^2$).
Let $\sigma_{i}'$ be $\sigma_{i}$ naturally extended on 
$R$, that is, 
two tuples from $R$ are equivalent if their 
$i$-th coordinates are equivalent modulo $\sigma_{i}$.
Let $E_{i}$ be the block of $\sigma_{i}'$ corresponding to $B_{i}$.
Since $R$ is subdirect, 
$\alg R/\sigma_{i}' \cong \algA_{i}/\sigma_{i}$.
Therefore, $E_{i}$ is a PC subuniverse of $\alg R$.
Note that 
$(R\cap (B_{1}\times\dots\times B_{n}))=E_{1}\cap\dots\cap E_{n}$.
By Corollary \ref{IntersectionOfPCisPC}, 
$E_{1}\cap\dots\cap E_{n}$ is a PC subuniverse.
\end{proof}

\begin{lem}\label{PCImpliesForTwo}
Suppose $R \le_{sd} \alg A_1\times\alg A_2$,
$\algA_1$ has no nontrivial central subuniverses,
$B_2$ is a PC subuniverse of $\alg A_2$,
and 
$B_1 = \proj_{1}(R\cap(A_{1}\times B_{2}))$.
Then $B_1$ is a PC subuniverse of $\algA_{1}$.
\end{lem}
\begin{proof}
By the definition of 
a PC subuniverse, 
$B_2$ is a block of a congruence 
$\sigma_2$ 
such that 
$\algA_{2}/\sigma_{2}\cong \alg D_{1}\times\dots\times\alg D_{k}$
for 
PC algebras $\alg D_1,\dots,\alg D_{k}$.
Consider the natural mapping 
$\psi:A_{2}\to D_{1}\times \dots\times D_{k}$
and the relation 
$R'\le \alg A_{1}\times \alg D_{1}\times\dots\times\alg D_{k}$ 
defined by 
$R'=\{(a_{1},\psi(a_{2}))\mid (a_1,a_2)\in R\}$.
Since $R'$ is subdirect, 
by Lemma~\ref{PCRelationsLem}
it can be represented 
by binary constraints from the first coordinate
to an $i$-th coordinate such that 
the $i$-th coordinate is uniquely-determined by the first,
and by bijective binary constraints between pairs of coordinates
other than first.
Note that 
each constraint from the first coordinate 
to an $i$-th coordinate is as follows.
There exists a congruence $\sigma$ on $\algA_{1}$ 
such that $\algA_{1}/\sigma$ is a PC algebra isomorphic to 
$\alg D_{i-1}$, then the constraint  
assigns to all elements of each block of $\sigma$ 
the corresponding element of $D_{i-1}$.

Suppose $\psi(B_{2}) = \{(b_{1},\dots,b_{k})\}$.
To define $B_{1}$ we need to fix all the coordinates of 
$R'$ other then first to $b_{1},\dots,b_{k}$.
Then $B_{1}$ is either empty, or an
intersection of blocks of congruences $\sigma$ 
such that $\algA_{1}/\sigma$ is a PC algebra without BACP.
Since each such block is a PC subuniverse,
by Corollary~\ref{IntersectionOfPCisPC}, 
$B_1$ is a PC subuniverse.
\end{proof}

\begin{lem}\label{PCImplies}
Suppose $R \le_{sd} \alg A_{1}\times\dots\times \alg A_{n}$,
$\algA_{1}$ has no nontrivial central subuniverses, 
$B_{i}$ is a PC subuniverse of $\algA_{i}$ for every $i\in[n]$,
and 
$B = \proj_{1}(R\cap(B_{1}\times\dots \times B_{n}))$.
Then $B$ is a PC subuniverse of $\algA_{1}$.
\end{lem}


\begin{proof}
Let 
$\proj_{2,\dots,n}(R) = D$
and $C=D\cap (B_{2}\times\dots\times B_n)$.
The relation $R$ can be viewed as a binary subdirect relation 
$R\le_{sd} \alg A_{1}\times \alg D$.
By Lemma~\ref{PCrestrictionOfRel},
$C\le_{PC} \alg D$.
By Lemma~\ref{PCImpliesForTwo},
$\proj_{1}(R\cap (A_{1}\times C))$ is a PC 
subuniverse of $\algA_{1}$.
Then $B$ is the intersection of 
$\proj_{1}(R\cap (A_{1}\times C))$ and $B_1$, 
and also a PC subuniverse by Corollary~\ref{IntersectionOfPCisPC}.
\end{proof}

\begin{lem}\label{PCLessThanThree}
Suppose
$B_{i}$ is a PC subuniverse of $\algA_{i}$
for $i\in[n]$, $n\ge 3$.
Then there does not exist
a $(B_{1},\dots,B_{n})$-essential relation 
$R\le_{sd} \algA_{1}\times \dots \times \algA_{n}$.
\end{lem}

\begin{proof}
Assume the contrary. By 
Lemma \ref{PCrestrictionOfRel}, 
each relation $R_{i}:= R\cap (A_{1}\times\dots\times A_{i-1}\times B_{i}\times
A_{i+1}\times\dots\times A_{n})$ is a PC subuniverse of $\alg R$.
Since $R$ is 
$(B_{1},\dots,B_{n})$-essential, 
$R_{1}\cap\dots\cap R_{n} = \varnothing$ 
and $R_{i}\cap R_{j}\neq \varnothing$ for all $i,j\in [n]$.
This contradicts 
Corollary \ref{EmptyPCIntersection}.
\end{proof}


\subsection{Common properties}

In this section we will prove statements formulated in 
Section~\ref{StrongSubalgebrasSection}.

\begin{lem}\label{GetBACenterFromPower}
Suppose $R$ is a nontrivial strong subuniverse of $\algA_{1}\times \dots\times \algA_{n}$ of type $\mathcal T\neq PC$.
Then there exists $i\in[n]$ such that 
$\algA_{i}$ has a nontrivial subuniverse of type $\mathcal T$.
\end{lem}
\begin{proof}
We prove by induction on the arity of $R$. For $n=1$ it is trivial.
If 
$\proj_{1}(R)\neq A_{1}$ then
by Lemmas~\ref{AbsImplies} and \ref{CenterImplies} this projection is a subuniverse of $\algA_1$ of type 
$\mathcal T$.

Otherwise, we choose any element $a_1\in A_{1}$
such that $R$ does not contain all tuples starting with $a_1$.
Then we consider $R' = \{(a_2,\ldots,a_{n})\mid (a_1,a_2,\ldots,a_n)\in R\}$, which, by Lemmas~\ref{AbsImplies} and \ref{CenterImplies}, is
a nontrivial subuniverse of $\algA_{2}\times\dots \times \algA_{n}$ 
of type $\mathcal T$.
It remains to apply the inductive assumption.
\end{proof}




\begin{LEMCBTNonAbsorbing}
Suppose $B$ is a nontrivial projective subuniverse of a finite idempotent algebra $\algA$, 
and $B$ is not a binary absorbing subuniverse.
Then there exists an essentially unary algebra 
$\alg U\in \HS(\algA)$ of size at least 2.
\end{LEMCBTNonAbsorbing}

\begin{proof}
By Lemma~\ref{ProjectiveSubCharacterization}, $R_{n} = A^{n}\setminus (A\setminus B)^{n}$ is an invariant of $\algA$ for every $n$.
Since $B$ is not a binary absorbing subuniverse, 
by Lemma~\ref{NoEssential} there exists a binary $B$-essential relation 
$R\in\Inv(\algA)$.
Put $R' = R\cap R_{2}$, 
$D = \proj_{1}(R')$, and 
$ S(x,y) = \exists x'\exists y'\; R'(x,x')\wedge R'(y,y')\wedge R_{2}(x',y')\wedge R_{2}(x,y).$
It is not hard to check that
$S = ((D\cap B)\times (D\setminus B))\cup 
((D\setminus B)\times (D\cap B))$.

Then $\sigma(x,y) = \exists z \; S(x,z)\wedge S(y,z)$
is the congruence on $\alg D$ having just two blocks
$(D\cap B)$ and $(D\setminus B)$.
We claim that the algebra
$\alg D/\sigma\in \HS(\algA)$ is essentially unary.

Since $B$ is a projective subuniverse, 
every $n$-ary operation $f$ of $\algA$ has a coordinate
$i\in[n]$
such that 
$f(b_{1},\dots,b_{n})\in B$ whenever $b_{i}\in B$.
If we restrict $f$ to $D$ (denote it by $f|_{D}$), 
then 
$f|_{D}(b_{1},\dots,b_{n})\in D\cap B$
whenever $b_{i}\in D\cap B$.
Since $f|_{D}$ preserves $S$, 
we also have
$f|_{D}(b_{1},\dots,b_{n})\in D\setminus B$
whenever $b_{i}\in D\setminus B$.
Thus, $f|_{D}$ is the $i$-th projection modulo $\sigma$.
Hence $\alg D/\sigma$ is essentially unary.
\end{proof}


\begin{THMCommonPropertiesThm}
Suppose 
$\alg R \le_{sd} \algA_{1}\times\dots\times \algA_{n}$,
$n\ge 2$,
$\algA_{1}, \dots,\algA_n$ are finite idempotent algebras, 
and 
$\alg B_{i}\le_{\mathcal T}\algA_{i}$ 
for every $i\in[n]$.
Then 
\begin{enumerate}
\item[(1)]
$(R\cap (B_{1}\times\dots\times B_{n}))\le_{\mathcal T} \alg R$;

\item[(2)]
if 
$\mathcal T\neq PC$ or 
$\algA_{1}$ has no nontrivial central subuniverses then 

$\proj_{1}(R\cap (B_{1}\times\dots \times B_{n}))\le_{\mathcal T}\algA_{1}$;


\item[(3)] 
if $R$ is $(B_{1},\dots,B_{n})$-essential
then 
$\mathcal T\in\{C,PC\}$  and $n=2$.
\end{enumerate}
\end{THMCommonPropertiesThm}
\begin{proof}
(1).
It follows from 
Corollary~\ref{RIntersectionBACons}, 
Corollary~\ref{RIntersectionCenterCons}, and 
Lemma~\ref{PCrestrictionOfRel}.

(2). For $\mathcal T = BA(t)$ it follows from 
Corollary \ref{AbsImpliesCons}, 
for $\mathcal T = C$ it follows 
from Corollary \ref{CenterImpliesCons},
for $\mathcal T = PC$ it follows 
from Lemma \ref{PCImplies}.

(3). For $\mathcal T = BA(t)$ we 
just apply the term operation $t$ to a tuple from 
$R\cap (A_{1}\times B_{2}\times\dots\times B_{n})$ and 
a tuple from $R\cap (B_{1}\times\dots\times B_{n-1}\times A_{n})$
to get a tuple from 
$R\cap (B_{1}\times\dots\times B_{n})$, which gives a contradiction.
For $\mathcal T = C$ it follows from
Lemma \ref{CenterLessThanThree}, 
for $\mathcal T = PC$ it follows from
Lemma \ref{PCLessThanThree}.
\end{proof}


Types of strong subalgebras are called \emph{similar} if 
they are the same or if they are
BA types (with probably different term operations).

\begin{lem}\label{IntersectionOfTwoSubuniverses}
Suppose 
$B_{1}$ and $B_{2}$
are nonempty strong subuniverses of a finite idempotent algebra $\algA$ and 
$B_{1}\cap B_{2} = \varnothing$.
Then $B_{1}$ and $B_{2}$ are of similar types.
\end{lem}

\begin{proof}
Assume the converse. Consider two cases.

Case 1. $B_{1}$ is a PC subuniverse, $B_{2}$ is a BA/central subuniverse.
Let $B_{1}$ be a block of a congruence $\sigma$ 
such that 
$\algA/\sigma\cong \alg D_{1}\times \dots\times\alg D_{s}$
(from the definition of a PC subuniverse).
By Corollary~\ref{AbsorptionQuotient}/Lemma~\ref{CenterQuotient}
the set $\{b/\sigma\mid b\in B_{2}\}$ 
is a BA/central subuniverse of $\algA/\sigma$.
Since $B_{1}\cap B_{2}\neq\varnothing$, 
this subuniverse is nontrivial.
By Lemma~\ref{GetBACenterFromPower}
there exists $i\in[s]$ such that $\alg D_{i}$ has a nontrivial BA/central
subuniverse, which contradicts 
the properties of $\alg D_{i}$.


Case 2. $B_{1}$ is a binary absorbing subuniverse, 
$B_{2}$ is a central subuniverse.
Let $e\in B_{1}$ be chosen so that 
the set $\Sg_{\algA}(B_{2}\cup\{e\})$ is inclusion minimal.
By Lemma \ref{NoEEForCenter} we have $\Sg_{\algA}
\begin{pmatrix}
e& B_{2} \\B_{2} & e 
\end{pmatrix}\cap B_{1}^{2}=\varnothing$.
Applying 
the binary absorbing term operation to 
$(e,b)$ and $(b,e)$ for some $b\in B_{2}$,
we get a tuple 
from $\Sg_{\algA}
\begin{pmatrix}
e& B_{2} \\B_{2} & e 
\end{pmatrix}\cap B_{1}^{2}$. Contradiction.
\end{proof}

\begin{LEMPCBsub}
Suppose $\algA$ is a finite idempotent algebra, 
$B_1$ and $B_{2}$ are subuniverses of $\algA$ of types $\mathcal T_{1}$
and $\mathcal T_{2}$, respectively.
Then $B_{1}\cap B_{2}$ is strong subuniverse of $\algB_{2}$ of type $\mathcal T_{1}$.
\end{LEMPCBsub}

\begin{proof}
If $B_{1}$ or $B_{2}$ is empty, then the claim is trivial.
If $B_{1}$ is a BA or  central subuniverse 
then the claim follows from 
Lemmas~\ref{AbsImplies} and \ref{CenterImplies}, respectively.

Assume that $B_{1}$ is a PC subuniverse of $\algA$.
If $B_{1}$ is full, then the claim is obvious.
Otherwise, 
$B_{1}$ is a block of a congruence $\sigma$ on $\algA$.
Assume that $B_{2}$ is not a 
PC subuniverse.
Any block $B$ of $\sigma$ is a PC subuniverse of $\algA$
and by Lemma~\ref{IntersectionOfTwoSubuniverses}
we have $B\cap B_{2}\neq \varnothing$.
Let $\sigma' =\sigma\cap (B_{2}\times B_{2})$.
Then 
$\algA/\sigma\cong \alg B_{2}/\sigma'$ 
and $B_{1}\cap B_{2}$ is a block of a congruence 
$\sigma'$, which means that 
$B_{1}\cap B_{2}$ is a PC subuniverse of 
$\alg B_{2}$.

If $B_{2}$ is also a PC subuniverse then
it follows from  Corollary~\ref{PCProperties}(2)
that 
$B_{1}\cap B_{2}$ is a PC subuniverse of $\alg B_{2}$.
\end{proof}

\begin{THMPCBint}
Suppose $\algA$ is a finite idempotent algebra, 
$B_{i}\le_{\mathcal T_{i}}\algA$
for every $i\in[n]$,
$n\ge 2$, 
$\bigcap_{i\in[n]}B_{i} = \varnothing$,
and 
$\bigcap_{i\in[n]\setminus\{j\}}B_{i} \neq  \varnothing$
for every $j\in[n]$.
Then 
one of the following conditions holds:
\begin{enumerate}
    \item[(1)] $n=2$ and $\mathcal T_{1} = \mathcal T_{2}\in\{C,PC\}$;
    \item[(2)] $\mathcal T_{1}, \dots, \mathcal T_{n}$ are binary absorbing types.
\end{enumerate}
\end{THMPCBint}

\begin{proof}
First, we prove by induction on $n$ that 
the types $\mathcal T_{1}, \dots, \mathcal T_{n}$ are similar.
For $n=2$ it follows from Lemma~\ref{IntersectionOfTwoSubuniverses}.
Assume that $n\ge 3$. 
Let $C_{i} = B_{i}\cap B_{n}$ for every $i\in[n-1]$.
By Lemma~\ref{PCBsub}, 
$C_{i}\le_{\mathcal T_{i}} \alg B_{n}$, 
and 
$C_{1},\dots, C_{n-1}$ satisfy all the assumptions 
of this theorem. 
By the inductive assumption 
$\mathcal T_{1},\dots,\mathcal T_{n-1}$ are similar types.
In the same way we can prove that 
$\mathcal T_{2},\dots,\mathcal T_{n}$ are similar types. 
Hence, all the types  $\mathcal T_{1},\dots,\mathcal T_{n}$ are similar.

If $\mathcal T_{1},\dots,\mathcal T_{n}$ are binary absorbing types, then 
we don't need anything else. 
If $\mathcal T_{1}=\dots = \mathcal T_{n}= PC$, 
then by Corollary \ref{EmptyPCIntersection} we have $n=2$.
Assume that  $\mathcal T_{1}=\dots = \mathcal T_{n}=C$.
Let $R$ be the $n$-ary relation 
consisting of all the constant tuples $(a,a,\dots,a)$.
Then
$R$ is a $(B_{1},\dots,B_{n})$-essential relation, 
which contradicts Lemma~\ref{CenterLessThanThree}.
\end{proof}


\subsection{Linear algebras and Maltsev operation}

An algebra $\alg A$ is called \emph{linear}
if there exists
an abelian group operation $\oplus$ on $A$ 
such that 
$(x_{1}\oplus x_{2} = x_{3}\oplus x_{4})\in \Inv(\algA)$.
If additionally 
$(A;\oplus)\cong 
(\mathbb Z_{p}^{s};+)$
then $\alg A$ is called \emph{$p$-linear}.
If a linear/$p$-linear algebra $\algA$ has a term operation 
$(x\ominus y\oplus z)$
then $\algA$ is called 
\emph{affine}/\emph{$p$-affine}.

An operation $m$ is called \emph{Maltsev} 
if $m(x,x,y) = m(y,x,x) = y$.

\begin{lem}\label{MaltsevImpliesAffine}
Suppose a linear algebra $\algA$ has 
a Maltsev term operation $m$.
Then $\algA$ is affine.
\end{lem}
\begin{proof}
Applying the term $m$ to the tuples
$(a_{1},a_{2},a_{1},a_{2})$,
$(a_{2},a_{2},a_{2},a_{2})$,
$(a_{3},a_{2},a_{2},a_{3})$
from the relation $x_{1}\oplus x_{2} = x_{3}\oplus x_{4}$
we obtain the tuple 
$(m(a_{1},a_{2},a_{3}),a_{2},a_{1},a_{3})$ from the same relation. 
Hence $m(a_{1},a_{2},a_{3})= a_{1}\ominus a_{2}\oplus a_{3}$ for all $a_{1},a_{2},a_{3}\in A$, which completes the proof.
\end{proof}

\begin{lem}\label{ReduceParallelogram}
Suppose 
$R\le \algA^{n}$, $1\le k<n$, and 
\begin{align*}
(a_{1},\dots,a_{k},b_{k+1},\ldots,b_{n})\in R,\\
(b_{1},\dots,b_{k},a_{k+1},\ldots,a_{n})\in R,\\
(b_{1},\dots,b_{k},b_{k+1},\ldots,b_{n})\in R,\\
(a_{1},\dots,a_{k},a_{k+1},\ldots,a_{n})\notin R.
\end{align*}
Then there exist $\rho\le \algA^{2}$
and $c_1,c_2,d_1,d_2\in A$ such that 
$(c_{1},c_{2})\notin \rho$ and 
$(c_{1},d_{2}),(d_{1},c_{2}),(d_{1},d_{2})\in \rho$.
\end{lem}
\begin{proof}
Let $R$ be a relation of the minimal arity satisfying the conditions of this lemma.
If $R$ is of arity 2 then we have the required property.
With out loss of generality, assume
that $k\ge 2$ (otherwise $n-k\ge 2$).
Define 
$$
R'(x_{2},\dots,x_{n}) = 
\exists x_{1}\; 
R(x_{1},\dots,x_{n}) \wedge 
R(x_{1},\dots,x_{k},b_{k+1},\dots,b_{n}).
$$
The projections of the three tuples from the statement onto the last $n-1$ variables are in $R'$, 
hence if
$(a_{2},\dots,a_{n})\notin R'$, 
then we derived a relation of a smaller arity 
satisfying the required property, 
which contradicts our assumption.
Assume that
$(a_{2},\dots,a_{n})\in R'$, then for some $c\in A$ 
we have 
$(c,a_{2},\dots,a_{n}),(c,a_{2},\dots,a_{k},b_{k+1},\dots,b_{n})\in R.$
Then for the relation $R''$
defined by 
$$
R''(x_{1},x_{k+1},\dots,x_{n}) = 
R(x_{1},a_{2},\dots,a_{k},x_{k+1},\dots,x_{n}).
$$
we have 
$(a_{1},b_{k+1},\ldots,b_{n}),
(c,a_{k+1},\ldots,a_{n}),
(c,b_{k+1},\ldots,b_{n})\in R''$,
and 
$(a_{1},a_{k+1},\ldots,a_{n})\notin R''$.
Again we get a relation of a smaller arity with the required properties, 
which contradicts the assumption about the minimality.
\end{proof}

\begin{lem}\label{NoMaltsevImplies}
Suppose $\Clo(\algA)$ has no Maltsev term operation. 
Then there exist $\rho\le \algA^{2}$
and $a_1,a_2,b_1,b_2\in A$ such that 
$(a_{1},b_{2}),(b_{1},a_{2}),(b_{1},b_{2})\in \rho$
and $(a_{1},a_{2})\notin \rho$.
\end{lem}
\begin{proof}
Consider 
a matrix with 2 columns and $k=|A|^{2}$ rows containing all 
pairs $(a,b)\in A^{2}$ as rows. 
Let $\alpha$ and $\beta$ be the two columns of this matrix.

Let $S = \Sg_{\algA}(\{\alpha\beta,\beta\beta,\beta\alpha\})$.
If $\alpha\alpha\in S$, then 
there exists a term operation $t$
such that 
$t(\alpha,\beta,\beta) =
t(\beta,\beta,\alpha)=\alpha$, hence $t$ is 
a Maltsev term operation on $A$. 
Otherwise, we get a 
relation $S\le \alg A^{2k}$ containing 
$\alpha\beta,\beta\beta,\beta\alpha$ but not containing 
$\alpha\alpha$. It remains to apply Lemma~\ref{ReduceParallelogram}.
\end{proof}

\begin{lem}\label{NoMaltsevImplies2}
Suppose 
$\algA$ has no a Maltsev term. Then 
\begin{enumerate}
    \item $\algA$ has a proper subuniverses of size at least 2, or 
    \item there exists a central relation $R\le \algA^{2}$.
\end{enumerate}
\end{lem}
\begin{proof}
By Lemma~\ref{NoMaltsevImplies}
there exist $\rho\le \algA^{2}$
and $a_1,a_2,b_1,b_2\in A$ such that 
$(a_{1},b_{2}),(b_{1},a_{2}),(b_{1},b_{2})\in \rho$
and $(a_{1},a_{2})\notin \rho$.
If $\proj_{1}(\rho)\neq A$ or 
$\proj_{2}(\rho)\neq A$, then we are done.

Otherwise, let 
$S = b_{1}+\rho$. If $S = A$ then $\rho$ is a central relation with a center containing $b_{1}$.
Otherwise, $S\supseteq\{a_{2},b_{2}\}$ is a proper subuniverse of $\algA$ of size at least 2.
\end{proof}

\begin{lem}\label{FromStronglyRich}
Suppose $R\le \algA^{n}$ is a full-projective and 
uniquely-determined relation, and $n\ge 3$.
Then there exists a central relation $\rho\le \algA^{2}$ 
or $\algA$ is a linear algebra.
\end{lem}
\begin{proof}
Choose elements $b,c\in A$ such that 
$(c,b,\dots,b)\in R$. Later we denote 
$c$ by $0$.
Define a ternary relation $\oplus(x_{1}, x_{2},x_{3})$ 
by 
$$\exists y_{1}\exists y_{2}\;
R(x_{1},y_{1},b,\dots,b)\wedge 
R(x_{2},b,y_{2},b,\dots,b)\wedge 
R(x_{3},y_{1},y_{2},b,\dots,b).$$
We want to show that $\oplus$ represents a binary 
operation from an abelian group. 
Let us prove all the required properties. 

\textbf{$\oplus$ is an operation.}
Since $R$ is uniquely-determined, for any $x_{1}$ and $x_{2}$
we have unique choices of $y_{1}$ and $y_{2}$, and therefore 
we have a unique choice for $x_{3}$. Thus, it is an operation.
Later we will use it as an operation
and write 
$x_1\oplus x_2 = x_3$ instead of 
$\oplus(x_1,x_2,x_3)$.
Note that if $x_{1}=0$ then $y_{1} = b$ and therefore 
$x_{2} = x_{3}$, hence $0\oplus a = a$ for every $a\in A$.
Similarly, $0\oplus a = a$.

\textbf{$\oplus$ is commutative.}
Let $\sigma(x_{1},x_{2}) = 
\exists x_{3} \; \oplus (x_{1},x_{2},x_{3})\wedge \oplus (x_{2},x_{1},x_{3}).$
In other words, 
$\sigma$ is the set of all pairs $(a_{1},a_{2})$ 
such that 
$a_{1}\oplus a_{2} = a_{2}\oplus a_{1}$.
Since $0\oplus a = a\oplus 0$, 
$(0,a),(a,0)\in\sigma$ for any $a\in A$. Hence, 
if $\sigma$ is not full, then we derived a central relation.
If $\sigma$ is full then 
$\oplus$ is commutative.

\textbf{$\oplus$ is associative.}
Let $\delta$ be the set of all tuples $(a_{1},a_{2},a_{3})$
such that 
$(a_{1}\oplus a_{2})\oplus a_{3} = a_{1}\oplus (a_{2}\oplus a_{3})$
(we could define $\delta$ by a pp-formula).
Note that 
$(0,a,b),(a,0,b)\in\delta$ for any $a,b\in A$. 
Assume that $\delta$ is not full, then 
choose $(a_{1},a_{2},a_{3})\in A^{3}\setminus\delta$
and define a central relation by 
$\rho(x,y) = \delta(a_{1},x,y)$, where 0 will be in the center. 

Thus $\oplus$ is a commutative associative operation.
Define $\rho$ by 
$$\rho(x_{1},x_{2},x_{3},x_{4})=\exists y\; \oplus (x_{1},x_{2},y)\wedge \oplus (x_{3},x_{4},y).$$
Then $\rho$ is 
$x_{1}\oplus x_{2} = x_{3}\oplus x_{4}$,
and therefore $\algA$ is linear.
\end{proof}

\begin{lem}\label{GetpAffineFromAffine}
Suppose $\algA$ is a linear algebra. Then 
there exists a nontrivial equivalence relation 
$\sigma\le \algA^{2}$ or 
$\algA$ is a $p$-linear algebra for some $p$. 
\end{lem}

\begin{proof}
Let $\rho\le\algA^{4}$ be the relation defined by 
$x_{1}\oplus x_{2} = x_{3}\oplus x_{4}$, where 
$(A;\oplus)$ is the abelian group from the definition of 
a linear algebra $\algA$.

Assume that two elements of this group have different orders $m$ and $n$, where $m<n$.
Let $\delta$ be the set of all pairs $(a,b)$ 
such that 
$\underbrace{a\oplus \dots\oplus a}_{m} = \underbrace{b\oplus \dots\oplus b}_{m}$.
Obviously, $\delta$ is pp-definable over $\rho$ and constant relations, and $\delta$ is a nontrivial equivalence relation on $A$.

If all elements of $(A;\oplus)$ have the same order then 
$(A;\oplus)\cong (\mathbb Z_{p}^{s};+)$
for some prime number $p$ and integer $s$,
and $\algA$ is $p$-linear.
\end{proof}

\subsection{Existence of a strong subalgebra}

\begin{lem}\label{FullProjective}
Suppose $C$ is a clone on a set $A$ and $R\in\Inv(C)$ is a relation of the minimal arity such that 
$\Pol(R)$ is not the set of all operations.
Then $R$ is full-projective.
\end{lem}
\begin{proof}
Let $R$ be of an arity $k$. Assume that $R' = \proj_{I}(R)$ is not full for some set $I\subsetneq [k]$.
Since $\Pol(R')$ should contain all operations, 
Lemma~\ref{PolIsFull} implies that 
$R'$ can be represented as a conjunction of the equality relations.
Thus, for some $i,j\in I$ 
the relation 
$\proj_{i,j}(R)$ is the equality relation.
Hence, 
$\proj_{[k]\setminus\{i\}}(R)$ is a
relation of a smaller arity
such that $\Pol(\proj_{[k]\setminus\{i\}}(R))=\Pol(R)$. Contradiction.
\end{proof}


\begin{lem}\label{ExistsStrongSubalgebraLEM}
Suppose $\algA$ is a finite idempotent algebra, then 
\begin{enumerate}
    \item there exists a nontrivial congruence on $\algA$, or
    \item there exists a central relation $R\le \algA^{2}$, or
    \item $\algA$ is polynomially complete, or 
    \item $\algA$ is $p$-affine.
\end{enumerate}
\end{lem}
\begin{proof}
Consider the clone $C$ generated from the operations of $\algA$ 
and all constant operations. 
If $C$ is the clone of all operations then 
$\algA$ is polynomially complete (case 3).

Otherwise, choose a relation $R\in \Inv(C)$ of the minimal arity
such that $\Pol(R)$ is not the clone of all operations.
By Lemma~\ref{FullProjective}, $R$ is full-projective.
Since $R$ is preserved by all constants, it cannot be unary.
By Lemma~\ref{CoolSubdirectImplies} 
$\Inv(\algA)$ contains a nontrivial equivalence relation (case 1) or a central relation (case 2),
or the relation $R$ is uniquely-determined, which is the only remaining case. 
Since $R$ is preserved by constants, it should contain all constant tuples. 
If $R$ is binary, then $R$ is the equality relation, 
which contradicts the fact that 
$\Pol(R)$ is not the clone of all operations.
Thus, 
$R$ is of arity at least 3. Then by 
Lemma~\ref{FromStronglyRich}
we either get a central relation (case 2), or
$\algA$ is a linear algebra.
Then, by Lemma~\ref{GetpAffineFromAffine}, we either get a nontrivial equivalence congruence (case 1), or 
$\algA$ is $p$-linear.
If $\algA$ has a Maltsev term, then by Lemma~\ref{MaltsevImpliesAffine}, 
$\algA$ is $p$-affine (case 4).
Otherwise, by Lemma~\ref{NoMaltsevImplies2}
there exists a central relation $R\le \algA^{2}$ (case 2),
or there exists 
$B\le \algA$ with $1<|B|<|A|$.
In the later case let
$R'(x_1,\dots,x_{n-1}) = 
\exists x_{n} \; R(x_{1},\dots,x_{n})\wedge B(x_{n})$.
Since $R$ is full-projective and uniquely-determined, $R'$ 
contains $|B|\cdot |A|^{n-2}$ tuples.
Therefore, it is not uniquely-determined, but it is still full-projective.
Then by Lemma~\ref{CoolSubdirectImplies} 
$\Inv(\algA)$ contains a central relation (case 2) or a nontrivial equivalence relation (case 1).
\end{proof}

\begin{THMExistsStrongSubalgebraTHM}
Every finite idempotent algebra $\algA$ of size at least 2 has
\begin{enumerate}
\item[(1)] a nontrivial binary absorbing subuniverse, or
\item[(2)] a nontrivial central subuniverse, or
\item[(3)] a nontrivial PC subuniverse, or
\item[(4)] a congruence $\sigma$ such that 
$\algA/\sigma$ is $p$-affine, or
\item[(5)] a nontrivial projective subuniverse.
\end{enumerate}
\end{THMExistsStrongSubalgebraTHM}

\begin{proof}
We prove by induction on the size of $\algA$.

Assume that there exists a nontrivial congruence on $\algA$.
Let $\delta$ be a maximal congruence on $\algA$.
By the inductive assumption $\algA/\delta$ satisfies one of the five conditions of the theorem.
In cases (1), (2) and (5) we have a nontrivial BA/central/projective subuniverse $B$ of $\algA/\delta$.
Then by Corollaries~\ref{AbsorptionQuotient}, \ref{CenterQuotientTwo}
and Lemma~\ref{ProjectiveQuotient}, 
$\bigcup_{E\in B} E$ 
is a BA/central/projective subuniverse of $\algA$.
In case (3), $\algA/\delta$ should be a PC algebra
and any block of $\delta$ is a nontrivial PC subuniverse of $\algA$.
In case (4) the congruence $\sigma$ should be trivial, and therefore
$\algA/\delta$ is a $p$-affine algebra.

Assume that $\algA$ has no nontrivial congruences.
By Lemma~\ref{ExistsStrongSubalgebraLEM} 
one of the cases 2-4 of this lemma 
should hold.
In case 2, Theorem ~\ref{CenterImpliesTHM} implies 
that
$\algA$ has a nontrivial binary absorbing subuniverse, or 
a nontrivial central subuniverse, or a nontrivial projective subuniverse, that is, cases (1), (2), or (5).
In cases 3 and 4 we obtain cases (3) and (4) of this theorem.
\end{proof}

\bibliography{refs}

\begin{thebibliography}{10}

\bibitem{DecidingAbsorption}
Libor Barto and Alexandr Kazda.
\newblock (2016).
\newblock Deciding absorption.
\newblock {\em International Journal of Algebra and Computation},
  26(05):1033--1060.

\bibitem{cyclicterms}
Libor Barto and Marcin Kozik.
\newblock (February 2012).
\newblock {Absorbing Subalgebras, Cyclic Terms, and the Constraint Satisfaction
  Problem}.
\newblock {\em {Logical Methods in Computer Science}}, {Volume 8, Issue 1}.

\bibitem{bartokozikboundedwidth}
Libor Barto and Marcin Kozik.
\newblock (2014).
\newblock Constraint satisfaction problems solvable by local consistency
  methods.
\newblock {\em Journal of the ACM (JACM)}, 61(1):1--19.

\bibitem{barto2017absorption}
Libor Barto and Marcin Kozik.
\newblock (2017).
\newblock Absorption in universal algebra and csp.

\bibitem{bartopolymorphisms}
Libor Barto, Andrei Krokhin, and Ross Willard.
\newblock (2017).
\newblock Polymorphisms, and how to use them.
\newblock Preprint.

\bibitem{bergman2011universal}
Clifford Bergman.
\newblock (2011).
\newblock {\em Universal algebra: Fundamentals and selected topics}.
\newblock CRC Press.

\bibitem{bond}
V.~G. Bodnarchuk, L.~A. Kaluzhnin, V.~N. Kotov, and B.~A. Romov.
\newblock (1969).
\newblock Galois theory for post algebras {parts I and II}.
\newblock {\em Cybernetics}, (5):243--252, 531--539.

\bibitem{bulatov2009bounded}
Andrei Bulatov.
\newblock (2009).
\newblock Bounded relational width.
\newblock {\em Unpublished manuscript}.

\bibitem{CSPconjecture}
Andrei Bulatov, Peter Jeavons, and Andrei Krokhin.
\newblock (March 2005).
\newblock Classifying the complexity of constraints using finite algebras.
\newblock {\em SIAM J. Comput.}, 34(3):720--742.

\bibitem{BulatovProofCSP}
Andrei~A. Bulatov.
\newblock (2017).
\newblock A dichotomy theorem for nonuniform csps.
\newblock {\em CoRR}, abs/1703.03021.

\bibitem{BulatovAboutCSP}
Andrei~A. Bulatov and Matthew~A. Valeriote.
\newblock (2008).
\newblock Recent results on the algebraic approach to the csp.
\newblock In Nadia Creignou, PhokionG. Kolaitis, and Heribert Vollmer, editors,
  {\em Complexity of Constraints}, volume 5250 of {\em Lecture Notes in
  Computer Science}, pages 68--92. Springer Berlin Heidelberg.

\bibitem{Num4}
Martin~C. Cooper.
\newblock (1994).
\newblock Characterising tractable constraints.
\newblock {\em Artificial Intelligence}, 65(2):347--361.

\bibitem{FederVardi}
Tom\'{a}s Feder and Moshe~Y. Vardi.
\newblock (February 1999).
\newblock The computational structure of monotone monadic snp and constraint
  satisfaction: A study through datalog and group theory.
\newblock {\em SIAM J. Comput.}, 28(1):57--104.

\bibitem{freese1987commutator}
Ralph Freese and Ralph McKenzie.
\newblock (1987).
\newblock {\em Commutator theory for congruence modular varieties}, volume 125.
\newblock CUP Archive.

\bibitem{geiger1968closed}
David Geiger.
\newblock (1968).
\newblock Closed systems of functions and predicates.
\newblock {\em Pacific journal of mathematics}, 27(1):95--100.

\bibitem{istinger1979characterization}
M~Istinger and HK~Kaiser.
\newblock (1979).
\newblock A characterization of polynomially complete algebras.
\newblock {\em Journal of Algebra}, 56(1):103--110.

\bibitem{jeavons1998algebraic}
Peter Jeavons.
\newblock (1998).
\newblock On the algebraic structure of combinatorial problems.
\newblock {\em Theoretical Computer Science}, 200(1-2):185--204.

\bibitem{Num20}
Peter Jeavons, David Cohen, and Marc Gyssens.
\newblock (July 1997).
\newblock Closure properties of constraints.
\newblock {\em J. ACM}, 44(4):527--548.

\bibitem{Num22}
Peter~G. Jeavons and Martin~C. Cooper.
\newblock (1995).
\newblock Tractable constraints on ordered domains.
\newblock {\em Artificial Intelligence}, 79(2):327--339.

\bibitem{Num23}
Lefteris~M. Kirousis.
\newblock (1993).
\newblock Fast parallel constraint satisfaction.
\newblock {\em Artificial Intelligence}, 64(1):147--160.

\bibitem{kozik2016weak}
Marcin Kozik.
\newblock (2016).
\newblock Weak consistency notions for all the csps of bounded width.
\newblock In {\em 2016 31st Annual ACM/IEEE Symposium on Logic in Computer
  Science (LICS)}, pages 1--9. IEEE.

\bibitem{lau}
D.~Lau.
\newblock (2006).
\newblock {\em Function algebras on finite sets}.
\newblock Springer.

\bibitem{lausch2000algebra}
Hans Lausch and Wilfred Nobauer.
\newblock (2000).
\newblock {\em Algebra of polynomials}, volume~5.
\newblock Elsevier.

\bibitem{miklos}
M.~Mar{\'o}ti and R.~Mckenzie.
\newblock (2008).
\newblock Existence theorems for weakly symmetric operations.
\newblock {\em Algebra universalis}, 59(3--4):463--489.

\bibitem{Post}
E.~L. Post.
\newblock (1941).
\newblock {\em The {T}wo-{V}alued {I}terative {S}ystems of {M}athematical
  {L}ogic}.
\newblock Annals of Mathematics Studies, no. 5. Princeton University Press,
  Princeton, N. J.

\bibitem{rosmax}
I.~Rosenberg.
\newblock (1970).
\newblock \"uber die funktionale vollst\"andigkeit in den mehrwertigen logiken.
\newblock {\em Rozpravy \v{C}eskoslovenske Akad. V\v{e}d., Ser. Math. Nat.
  Sci.}, 80:3--93.

\bibitem{Schaefer}
Thomas~J. Schaefer.
\newblock (1978).
\newblock The complexity of satisfiability problems.
\newblock In {\em Proceedings of the Tenth Annual ACM Symposium on Theory of
  Computing}, STOC '78, pages 216--226, New York, NY, USA. ACM.

\bibitem{KeyRelations}
Dmitriy Zhuk.
\newblock (2017).
\newblock Key (critical) relations preserved by a weak near-unanimity function.
\newblock {\em Algebra Universalis}, 77(2):191--235.

\bibitem{MyProofCSP}
Dmitriy Zhuk.
\newblock (2017).
\newblock A proof of csp dichotomy conjecture.
\newblock {\em CoRR}, abs/1704.01914.

\bibitem{ZhukFOCSCSPPaper}
Dmitriy Zhuk.
\newblock (2017).
\newblock A proof of {CSP} dichotomy conjecture.
\newblock In {\em 58th {IEEE} Annual Symposium on Foundations of Computer
  Science, {FOCS} 2017, Berkeley, CA, USA, October 15-17, 2017}, pages
  331--342.

\end{thebibliography}

\appendix

\end{document}